\documentclass{cta-author}

\newtheorem{theorem}{Theorem}{}
\newtheorem{example}{Example}{}
{}
\newtheorem{defi}{Definition}{}
{}
\usepackage{epstopdf}

\begin{document}

\title{A General Approach for Construction of Deterministic Compressive Sensing Matrices}

\author{\au{Mohamad Mahdi Mohades$^{1}$}, \au{Mohamad Hossein Kahaei$^{1\corr}$}}

\address{\add{1}{1Department of Electronics and Communication Engineering,  Iran University of Science and Technology, Tehran,
       Iran}
\email{kahaei@iust.ac.ir}}

\begin{abstract}
In this paper, deterministic construction of measurement matrices in Compressive Sensing (CS) is considered.
First, by employing the column replacement concept, a theorem for construction of large minimum distance linear
codes containing all-one codewords is proposed. Then, by applying an existing theorem over these linear codes,
deterministic sensing matrices are constructed. To evaluate this procedure, two examples of constructed sensing
matrices are presented. The first example contains a matrix of size ${{p}^{2}}\times {{p}^{3}}$ and
coherence ${1}/{p}\;$, and the second one comprises a matrix with the size $p\left( p-1 \right)\times {{p}^{3}}$ and
coherence ${1}/{\left( p-1 \right)}\;$, where $p$ is a prime integer. Based on the Welch bound, both examples
asymptotically achieve optimal results.  Moreover, by presenting a new theorem, the column replacement is used
for resizing any sensing matrix to a greater-size sensing matrix whose coherence is calculated. Then, using
an example, the outperformance of the proposed method is compared to a well-known method.
Simulation results show the satisfying performance of the column replacement method either in created or
resized sensing matrices.
\end{abstract}

\maketitle

\section{Introduction}\label{sec:Introduction}

Compressive Sampling (CS) has widely been used to dramatically reduce the sampling frequency below the Nyquist rate \cite{Donoho_1}. This is effectively performed based on the sparse property of a signal vector ${{\mathbf{x}}_{n\times 1}}\in {{\mathbb{R}}^{n}}$ whose sparsity order, $k$, indicates the number of its non-zero entries. Then, the CS problem is illustrated by ${{\mathbf{y}}_{m\times 1}}={{\mathbf{A}}_{m\times n}}{{\mathbf{x}}_{n\times 1}}$ where ${{\mathbf{y}}_{m\times 1}}\in {{\mathbb{R}}^{m}}$ shows the measurement vector, ${{\mathbf{A}}_{m\times n}}$ is the sampling matrix (or sensing matrix), and $m\ll n$. This sampling is worthwhile provided that the original signal ${{\mathbf{x}}_{n\times 1}}$ can be uniquely and exactly reconstructed. The sparsest solution to this problem using ${{l}_{0}}$-norm minimization has been investigated in \cite{Candes_2} which leads to an NP-hard problem \cite{Ge_3}. To simplify this issue, an ${{l}_{1}}$-norm minimization solution has been proposed in \cite{Candes_4} in which ${{\mathbf{A}}_{m\times n}}$ must satisfy the so-called Restricted Isometry Property (RIP). Then, ${{\mathbf{x}}_{n\times 1}}$ can be exactly and uniquely reconstructed by using ${{l}_{1}}$-norm minimization.

\begin{defi}
 The matrix ${{\mathbf{A}}_{m\times n}}$ satisfies the RIP of order $k$ with constant ${{\delta }_{k}}\in \left[ \begin{matrix} 0 & 1 \\ \end{matrix} \right)$, if for all $k$-sparse vectors ${{\mathbf{x}}_{n\times 1}}$ we have the next inequality \cite{Candes_4},
\end {defi}

\begin{equation}\label{Eqn:RIP}
1-{{\delta }_{k}}\le \frac{\left\| {{\mathbf{A}}_{m\times n}}{{\mathbf{x}}_{n\times 1}} \right\|_{2}^{2}}{\left\| {{\mathbf{x}}_{n\times 1}} \right\|_{2}^{2}}\le 1+{{\delta }_{k}}.
\end{equation}

As mentioned, the signal ${{\mathbf{x}}_{n\times 1}}$ sampled by the appropriate sampling matrix ${{\mathbf{A}}_{m\times n}}$ could be reconstructed through solving the related ${{l}_{1}}$-norm minimization problem. The latter problem may also be solved using the greedy algorithms such as the Orthogonal Matching Pursuit (OMP) algorithm or its modifications \cite{Tropp_5}.
At the beginning, Gaussian sensing matrices whose elements are independent and identically distributed (i.i.d.) were proposed which satisfy the RIP with a high probability \cite{Baraniuk_6}. For such matrices, the number of rows, columns, and the sparsity order are related to each other through the inequality $m\ge O\left( k\log n \right)$.
On the other hand, deterministic sensing matrices may be used. Even though, due to the huge number of $k$-sparse signals, evaluation of the RIP for a given deterministic sensing matrix is frustrating, such a matrix can still satisfy the RIP if its coherence is low enough. A deterministic sensing matrix ${{\mathbf{A}}_{m\times n}}$ of coherence ${{\mu }_{\mathbf{A}}}$ satisfies the RIP of order $k$ with ${{\delta }_{k}}\le {{\mu }_{\mathbf{A}}}\left( k-1 \right)$ where $k\le {1}/{{{\mu }_{\mathbf{A}}}}\;+1$ \cite{Bourgain_7}. The coherence of ${{\mathbf{A}}_{m\times n}}$, whose columns are defined as $\left[ {{\mathbf{a}}_{1}},\cdots ,{{\mathbf{a}}_{N}} \right]$, is given by

\begin{equation}\label{Eqn:Coherence}
{{\mu }_{\mathbf{A}}}=\underset{i\ne j}{\mathop{\max }}\,\frac{\left| \left\langle {{\mathbf{a}}_{i}},{{\mathbf{a}}_{j}} \right\rangle  \right|}{{{\left\| {{\mathbf{a}}_{{{i}_{{}}}}} \right\|}_{2}}.{{\left\| {{\mathbf{a}}_{j}} \right\|}_{2}}}\ \ \ \ \ \ \ 1\le i,j\le N.
\end{equation}

The minimum coherence of an arbitrary ${{\mathbf{A}}_{m\times n}}$ is provided by the Welch bound expressed as \cite{Welch_8}

\begin{equation}\label{Eqn:Welch}
{{\mu }_{\mathbf{A}}}\ge \sqrt{\frac{n-m}{m\left( n-1 \right)}}.
\end{equation}

where ${{\mu }_{\mathbf{A}}}\ge \sqrt{{1}/{m}\;}$ for $m\ll n$. This means that for deterministic sensing matrices, in terms of coherence, the RIP of sparsity $k=O\left( \sqrt{m} \right)$ is satisfied which is known as the square root bottleneck.
To construct deterministic sensing matrices with a low coherence, the DeVore's construction was firstly introduced \cite{DeVore_9}. In this structure, binary sensing matrices of size ${{p}^{2}}\times {{p}^{r+1}}$ and coherence of ${r}/{p}\;$ are used which satisfy the RIP of order $k<{p}/{r}\;+1$ where $r$ shows the degree of a polynomial with coefficients in Galois field ${{\mathbb{F}}_{p}}$. By considering polynomials over finite projective spaces, instead of the finite field ${{\mathbb{F}}_{p}}$, extension of the DeVore's construction  has led to lower coherence sensing matrices \cite{Mohades_10}.
The algebraic geometry codes \cite{Li_11} have also been utilized for deterministic sensing matrices construction.
Other binary structures can be seen in \cite{sasmal_24}, \cite{naidu_14} and \cite{naidu_25}.
In \cite{Mohades_15}, by combining two similar Reed-Solomon generator matrices through the tensor product and utilizing the generated codewords, a set of complex-valued deterministic sensing matrices with the size of ${{q}^{2}}\times {{q}^{3}}$ and coherence ${1}/{q}\;$, where $q$ is any prime power, has been proposed. Using the Paley graphs, a set of deterministic sensing matrices of size  ${\left( q+1 \right)}/{2}\;\times q$ with coherence less than twice the Welch bound has been developed \cite{Amini_16}. In \cite{Amini_18}, by employing BCH codes and creating large minimum distance codewords, some binary, bipolar and ternary deterministic sensing matrices are obtained.
In \cite{Amini_19}, to overcome the limitation on matrices’ sizes, some methods are developed to resize the sensing matrices for complex valued sensing matrices.
From the other aspect, we should note that the RIP is only a sufficient condition meaning that there might be some other matrices not satisfying this condition or even RIP-less, while can be used as sensing matrices  \cite{Xu_20}. In \cite{Calderbank_21}, some matrices which satisfy a weaker condition than the RIP, so-called statistical RIP, are proposed.
In this paper, to construct low coherence sensing matrices, the theorem of employing large minimum distance linear codes containing all-one codeword is first introduced \cite{Amini_19}. Then, we exhibit the procedure of employing the column replacement concept to create large minimum distance linear codes containing all-one codeword. Secondly, we present two examples in which the theorem of constructing sensing matrices \cite{Amini_19} is applied over our proposed matrices and we compute the coherence of the resulting matrices. Next, by comparing our results with some well-known sensing matrices, we exhibit the generality of the proposed method. These deterministic sensing matrices are of sizes ${{p}^{2}}\times {{p}^{3}}$ and $p\left( p-1 \right)\times {{p}^{3}}$ whose coherence are ${1}/{p}\;$and ${1}/{\left( p-1 \right)}\;$, respectively; where $p$ is a prime integer. The former example is quite similar to that of \cite{Mohades_15} for the prime integer $p$. About the second example, considering that in \cite{Amini_19} matrices of size $\left( {{p}^{2}}-1 \right)\times {{p}^{3}}$ and coherence ${1}/{\left( p-1 \right)}\;$ are suggested, our proposed structure outweighs that of \cite{Amini_19}  in terms of coherence owing to the fact that our method considers the same coherence with lesser rows in matrix construction.  It is notable that the mentioned examples are asymptotically optimal by the Welch bound. Finally, we utilize the column replacement method to resize the sensing matrices and calculate their coherence. By presenting an example, we show that for  the same size of sensing matrices, our procedure results in better sensing matrices in the sense of coherence  compared to the Kronecker product developed in \cite{Amini_19}.
The rest of paper is as follows. Section \ref{sec:main result} contains the main results. In Section \ref{Simulation_results}, using computer simulations, the performance of our sensing matrices is compared with some existing matrices. Section \ref{Conclusion} concludes the paper.

\section{Construction of Deterministic Sensing Matrices}\label{sec:main result}

To employ large minimum distance codes to generate low coherence deterministic sensing matrices, the following theorem applies.

\begin {theorem} \label {CS_Generation} \cite{Amini_19}
 Suppose that $\mathcal{C}\left[ N,k,p,{{d}_{\min }} \right]$ is a p-ary linear code over finite field ${{\mathbb{F}}_{p}}$ whose minimum distance is ${{d}_{min}}$ and all-one codeword belongs to the codewords' space. Now, assume that matrix ${{\mathbf{\tilde{A}}}_{N\times {{p}^{k-1}}}}$ is created by deploying each codeword as a column of this matrix in such a way that from any set of form $\left\{ \mathbf{a},\mathbf{a}\oplus {{\mathbf{1}}_{N\times 1}},...,\mathbf{a}\oplus {{(\mathbf{p-1})}_{N\times 1}} \right\}$ exactly one codeword is selected, where $\oplus $ stands for element to element summation over finite field ${{\mathbb{F}}_{p}}$. Then, we create ${{\mathbf{A}}_{N\times {{p}^{k-1}}}}$ as
\begin{equation}\label{Eqn:---}\mathbf{A}\text{=}\frac{1}{\sqrt{N}}{{\left[ {{e}^{j\frac{2\pi }{p}{{{\tilde{a}}}_{\alpha \beta }}}} \right]}_{\alpha ,\beta }}.
\end{equation}
 where ${{\tilde{a}}_{\alpha \beta }}$ is an element of ${{\mathbf{\tilde{A}}}_{N\times {{p}^{k-1}}}}$. Then, the coherence of $\mathbf{A}$ will be less than ${\left( p(p-1)N-{{p}^{2}}{{d}_{min}} \right)}/{2N}\;$.
\end {theorem}
This theorem enables us to think of low coherence sensing matrices provided that we can create large minimum distance matrix of codewords. We create such matrices based on column replacement of two matrices as follows.
\begin{defi}\cite{Colbourn_23} \label {Definition:Column_Replacement}
Consider the primary matrix $\mathbf{A}\in {{\text{R}}^{r\times m}}$, where $\mathbf{A}\text{=(}{{a}_{ij}}{{\text{)}}_{i,j}}$ and the pattern matrix $P\in {{\{1,...,m\}}^{N\times n}}$, where $\mathbf{P}\text{=(}{{p}_{ij}}{{\text{)}}_{i,j}}$. Then, the result of column replacement of  $\mathbf{A}$ in $\mathbf{P}$ is a matrix with the form of $\mathbf{C}=\text{B(}\mathbf{P}\text{)=(}{{b}_{ij}}\text{)}\in {{\text{R}}^{rN\times n}}$ , where ${{b}_{(\beta -1)r+s,\gamma }}={{a}_{s,{{p}_{\beta \gamma }}}}$ for $1\le \beta \le N$ , $1\le \gamma \le n$, and $1\le s\le r$.
\end {defi}
One can easily see that from \ref{Definition:Column_Replacement}, the following equation holds:
\begin{equation}\label{Eqn:CR_Deduction}
\mathbf{C}=\text{B(}\mathbf{P}\text{)=}\left( \text{B}\left( {{p}_{ij}} \right) \right).
\end{equation}

This means that column replacement is applied over each element of the pattern matrix. We use the column replacement to generate large minimum distance codes.

\subsection {Generating Large Minimum Distance Codes} \label{Large_Minimum_Distance_Codes}

We apply the column replacement method over matrices of code to propose a new matrix. It is proved  that the columns of this matrix are also codewords and their minimum distance is a function of the  primary and pattern matrices. Therefore, under some criteria, the minimum distance can be large enough to apply Theorem \ref{CS_Generation}. Using the column replacement method, we define and prove the following theorem to construct large minimum distance linear codes.

\begin {theorem} \label {Our_CR_Method}
Let ${{\mathbf{A}}_{N\times {{p}^{k}}}}$ and ${{\mathbf{P}}_{{N}'\times {{p}^{k{k}'}}}}$ be matrices whose columns are codewords of the linear codes $\mathcal{C}\left[ N,k,p,{{d}_{\min }} \right]$ and ${\mathcal{C}}'\left[ {N}',{k}',{{p}^{k}},{{{{d}'}}_{\min }} \right]$, respectively, where ${{d}_{\min }}$ and ${{{d}'}_{\min }}$ are the minimum distances of each defined linear code and $p$ is a prime integer. Moreover, suppose that both linear codes own all-one codeword. Then, the result of column replacement of ${{\mathbf{A}}_{N\times {{p}^{k}}}}$ in ${{\mathbf{P}}_{{N}'\times {{p}^{k{k}'}}}}$ will be the matrix ${{\mathbf{C}}_{N{N}'\times {{p}^{k{k}'}}}}=\text{B(}{{\mathbf{P}}_{{N}'\times {{p}^{k{k}'}}}}\text{)}$ whose columns are codewords of a linear code containing all-one codeword with minimum distance of $N{N}'-\left( \left( {N}'-{{{{d}'}}_{\min }} \right)N+{{{{d}'}}_{\min }}\left( N-{{d}_{\min }} \right) \right)$.
\end{theorem}

\begin {proof}  \label{proof_Our_CR_mehod}
We should note that it is natural to construct ${{\mathbf{P}}_{{N}'\times {{p}^{k{k}'}}}}$ over finite field ${{\mathbb{F}}_{{{p}^{k}}}}$ to correspond each element of this matrix to a column of matrix ${{\mathbf{A}}_{N\times {{p}^{k}}}}$.  First, we prove that the columns of ${{\mathbf{C}}_{N{N}'\times {{p}^{k{k}'}}}}$ are codewords. To do so, we must prove that any linear combination of every two columns will result in another column of this matrix. Thus, the columns belong to a linear subspace and are codewords. Suppose that ${{\mathbf{A}}_{N\times {{p}^{k}}}}$, ${{\mathbf{P}}_{{N}'\times {{p}^{k{k}'}}}}$, and ${{\mathbf{C}}_{N{N}'\times {{p}^{k{k}'}}}}$ have the following form
\begin{equation}\label{Eqn:Primary_Mat}
\resizebox{.95 \hsize}{!}
{$\begin{array}{l}
 {\bf A}_{N \times p^k }  = \left[ {{\bf a}_{N \times 1}^{0\alpha ^{(k - 1)}  + 0\alpha ^{(k - 2)}  +  \cdots  + 0\alpha ^1  + 0} , \cdots ,} \right. \\
 \;\;\;\;\;\;\;\;\;\;\;\;\;\;\;\;\;\;{\bf a}_{N \times 1}^{i_{k - 1} \alpha ^{k - 1}  + i_{k - 2} \alpha ^{k - 2}  +  \cdots  + i_1 \alpha ^1  + i_0 } , \cdots , \\
 \left. {\;\;\;\;\;\;\;\;\;\;\;\;\;\;\;\;\;\;{\bf a}_{N \times 1}^{\left( {p - 1} \right)\alpha ^{k - 1}  + \left( {p - 1} \right)\alpha ^{k - 2}  +  \cdots  + \left( {p - 1} \right)\alpha ^1  + \left( {p - 1} \right)} } \right] \end{array}$}
\end{equation}
\begin{equation}\label{Eqn:---}
{{\mathbf{P}}_{{N}'\times {{p}^{k{k}'}}}}=\left[ \mathbf{p}_{{N}'\times 1}^{1},\cdots ,\mathbf{p}_{{N}'\times 1}^{j},\cdots ,\mathbf{p}_{{N}'\times 1}^{{{p}^{k{k}'}}} \right]
\end{equation}
\begin{equation}\label{Eqn:---}
{{\mathbf{C}}_{N{N}'\times {{p}^{k{k}'}}}}=\left[ \mathbf{c}_{N{N}'\times 1}^{1},\cdots ,\mathbf{c}_{N{N}'\times 1}^{j},\cdots ,\mathbf{c}_{N{N}'\times 1}^{{{p}^{k{k}'}}} \right]
\end{equation}
where 
\\ $\mathbf{a}_{N\times 1}^{{{i}_{k-1}}{{\alpha }^{k-1}}+{{i}_{k-2}}{{\alpha }^{k-2}}+\cdots +{{i}_{1}}{{\alpha }^{1}}+{{i}_{0}}}={{i}_{k-1}}{{\mathbf{g}}_{k-1}}+{{i}_{k-2}}{{\mathbf{g}}_{k-2}}+\\ \cdots +{{i}_{1}}{{\mathbf{g}}_{1}}+{{i}_{0}}{{\mathbf{g}}_{0}}$, ${{\mathbf{G}}_{k\times N}}={{\left[ {{\mathbf{g}}_{0}},\cdots ,{{\mathbf{g}}_{k-1}} \right]}^{T}}$ is the generator matrix of codewords of ${{\mathbf{A}}_{N\times {{p}^{k}}}}$, ${{\mathbf{C}}_{N{N}'\times {{p}^{k{k}'}}}}=\text{B(}{{\mathbf{P}}_{{N}'\times {{p}^{k{k}'}}}}\text{)}$ or equivalently \\ \resizebox{1\hsize}{!}{${{\mathbf{C}}_{N{N}'\times {{p}^{k{k}'}}}}= \\ \left[ \text{B(}\mathbf{p}_{{N}'\times 1}^{1}\text{)},\cdots ,\text{B(}\mathbf{p}_{{N}'\times 1}^{j}\text{)},\cdots ,\text{B(}\mathbf{p}_{{N}'\times 1}^{{{p}^{k{k}'}}}\text{)} \right]$}, and $p_{i}^{j}$ belongs to the finite field ${{\mathbb{F}}_{{{p}^{k}}}}$ with  $1\le i\le {N}'$ and $1\le j\le {{p}^{k{k}'}}$. It is obvious that the elements of ${{\mathbf{C}}_{N{N}'\times {{p}^{k{k}'}}}}$ belong to ${{\mathbb{F}}_{p}}$ since the elements of ${{\mathbf{A}}_{N\times {{p}^{k}}}}$ belong to this field. Moreover, in definition of ${{\mathbf{A}}_{N\times {{p}^{k}}}}$ the superscript ${{i}_{k-1}}{{\alpha }^{k-1}}+{{i}_{k-2}}{{\alpha }^{k-2}}+\cdots +{{i}_{1}}{{\alpha }^{1}}+{{i}_{0}}$ is selected to corroborate the fact that the columns of ${{\mathbf{A}}_{N\times {{p}^{k}}}}$ correspond to the elements of ${{\mathbb{F}}_{{{p}^{k}}}}$ whose elements are in the form of ${{i}_{k-1}}{{\alpha }^{k-1}}+{{i}_{k-2}}{{\alpha }^{k-2}}+\cdots +{{i}_{1}}{{\alpha }^{1}}+{{i}_{0}}$, where $0\le {{i}_{k-1}},{{i}_{k-2}},\cdots ,{{i}_{1}},{{i}_{0}}\le p-1$. Furthermore, such correspondence means that if an element of ${{\mathbf{P}}_{{N}'\times {{p}^{k{k}'}}}}$ is ${{i}_{k-1}}{{\alpha }^{k-1}}+{{i}_{k-2}}{{\alpha }^{k-2}}+\cdots +{{i}_{1}}{{\alpha }^{1}}+{{i}_{0}}$, then the corresponding column of ${{\mathbf{A}}_{N\times {{p}^{k}}}}$is ${{i}_{k-1}}{{\mathbf{g}}_{k-1}}+{{i}_{k-2}}{{\mathbf{g}}_{k-2}}+\cdots +{{i}_{1}}{{\mathbf{g}}_{1}}+{{i}_{0}}{{\mathbf{g}}_{0}}$. Such relationship will let us pursue our proof much easier. Now, consider an arbitrary combination of two arbitrary columns of ${{\mathbf{C}}_{N{N}'\times {{p}^{k{k}'}}}}$ as
\begin{equation}\label{Eqn:Combination_2_Arb}
\mathbf{{c}''}={{\beta }_{1}}\mathbf{c}_{N{N}'\times 1}^{m}+{{\beta }_{2}}\mathbf{c}_{N{N}'\times 1}^{n}
\end{equation}
where ${{\beta }_{1}}$ and ${{\beta }_{2}}$ are any element of ${{\mathbb{F}}_{p}}$ and $\mathbf{c}_{N{N}'\times 1}^{m}$ and $\mathbf{c}_{N{N}'\times 1}^{n}$ are two columns of ${{\mathbf{C}}_{N{N}'\times {{p}^{k{k}'}}}}$. Next, we should prove that $\mathbf{{c}''}$ is also a column of ${{\mathbf{C}}_{N{N}'\times {{p}^{k{k}'}}}}$. Suppose that by applying the column replacement over two columns of ${{\mathbf{P}}_{{N}'\times {{p}^{k{k}'}}}}$, namely, $\mathbf{p}_{{N}'\times 1}^{m}$ and  $\mathbf{p}_{{N}'\times 1}^{n}$, $\mathbf{c}_{N{N}'\times 1}^{m}$ and $\mathbf{c}_{N{N}'\times 1}^{n}$ are obtained as
\begin{equation}\label{Eqn:Column_m}
\resizebox{.95 \hsize}{!}{$
\mathbf{c}_{NN' \times 1}^m  = \text{B}\left( {\mathbf{p}_{N' \times 1}^m } \right) = \text{B}\left[ {\begin{array}{*{20}c}
   {\varepsilon _{k - 1} \mathbf{g}_{k - 1}  +  \cdots  + \varepsilon _0 \mathbf{g}_0 }  \\
    \vdots   \\
   {\tau _{k - 1} \mathbf{g}_{k - 1}  +  \cdots  + \tau _0 \mathbf{g}_0 }  \\

 \end{array} } \right]$}\end{equation}

where B(.) presents the column replacement operator and  $\left\{ {\begin{array}{*{20}c}
   {0 \le \varepsilon _{k - 1} ,\varepsilon _{k - 2} , \cdots ,\varepsilon _1 ,\varepsilon _0  \le p - 1}  \\
   {0 \le \tau _{k - 1} ,\tau _{k - 2} , \cdots ,\tau _1 ,\tau _0  \le p - 1}  \\
\end{array}} \right.$.

\begin{equation}\label{Eqn:Column_n}
\mathbf{c}_{NN' \times 1}^n  = \text{B}\left[ {\begin{array}{*{20}c}
   {\varepsilon '_{k - 1} \mathbf{g}_{k - 1}  +  \cdots  + \varepsilon '_0 \mathbf{g}_0 }  \\
    \vdots   \\
   {\tau '_{k - 1} \mathbf{g}_{k - 1}  +  \cdots  + \tau '_0 \mathbf{g}_0 }  \\

 \end{array} } \right]
\end{equation}
\\
\\
where $\left\{ {\begin{array}{*{20}c}
   {0 \le \varepsilon '_{k - 1} ,\varepsilon '_{k - 2} , \cdots ,\varepsilon '_1 ,\varepsilon '_0  \le p - 1}  \\
   {0 \le \tau '_{k - 1} ,\tau '_{k - 2} , \cdots ,\tau '_1 ,\tau '_0  \le p - 1}  \\
\end{array}} \right.$.
\\
\\
Using (\ref{Eqn:Column_m}) and (\ref{Eqn:Column_n}) in (\ref{Eqn:Combination_2_Arb}), we get
\begin{equation}\label{Eqn:---}
\mathbf{{c}''}=\text{B}\left( {{\beta }_{1}}\mathbf{p}_{{N}'\times 1}^{m}+{{\beta }_{2}}\mathbf{p}_{{N}'\times 1}^{n} \right)
\end{equation}

Considering the fact that ${{\beta }_{1}}\mathbf{p}_{{N}'\times 1}^{m}+{{\beta }_{2}}\mathbf{p}_{{N}'\times 1}^{n}\in {{\mathbf{P}}_{{N}'\times {{p}^{k{k}'}}}}$, it is inferred that ${{\mathbf{C}}_{N{N}'\times {{p}^{k{k}'}}}}$ is a matrix whose columns are also codewords of a linear subspace.  Note that although we used the characteristic of ${{\beta }_{1}}\text{B(}\mathbf{p}_{{N}'\times 1}^{m}\text{)+}{{\beta }_{2}}\text{B(}\mathbf{p}_{{N}'\times 1}^{n}\text{)=B}\left( {{\beta }_{1}}\mathbf{p}_{{N}'\times 1}^{m}\text{+}{{\beta }_{2}}\mathbf{p}_{{N}'\times 1}^{n} \right)$ for ${{\beta }_{0}},{{\beta }_{1}}\in {{\mathbb{F}}_{p}}$, it can be applied to any ${{\beta }_{0}},{{\beta }_{1}}$ in the given finite field ${{\mathbb{F}}_{{{p}^{k}}}}$.
To analytically calculate the minimum distance of codewords included in ${{\mathbf{C}}_{N{N}'\times {{p}^{k{k}'}}}}$, we replace each column of ${{\mathbf{A}}_{N\times {{p}^{k}}}}$ instead of each element of ${{\mathbf{P}}_{{N}'\times {{p}^{k{k}'}}}}$.  In construction of each column of ${{\mathbf{C}}_{N{N}'\times {{p}^{k{k}'}}}}$, the worst case happens when: 1) two columns of ${{\mathbf{P}}_{{N}'\times {{p}^{k{k}'}}}}$ have exactly ${N}'-{{{d}'}_{\min }}$ elements in common and 2) the rest of elements of these two columns are chosen in such a way that the corresponding columns of ${{\mathbf{A}}_{N\times {{p}^{k}}}}$ have $N-{{d}_{\min }}$ elements in common. By such an assumption, the resulting columns will have $\left( {N}'-{{{{d}'}}_{\min }} \right)N+{{{d}'}_{\min }}\left( N-{{d}_{\min }} \right)$ common elements which is the worst case, and therefore, the minimum distance of codewords of ${{\mathbf{C}}_{N{N}'\times {{p}^{k{k}'}}}}$ is
\begin{equation}\label{Eqn:---}
\resizebox{0.9\hsize}{!}{$
d_{\min }^{\mathbf{C}}=N{N}'-\left( \left( {N}'-{{{{d}'}}_{\min }} \right)N+{{{{d}'}}_{\min }}\left( N-{{d}_{\min }} \right) \right)$}
\end{equation}
The last step is to prove that ${{\mathbf{C}}_{N{N}'\times {{p}^{k{k}'}}}}$ owns the all-one codeword. In the above theorem, it is assumed that ${{\mathbf{A}}_{N\times {{p}^{k}}}}$ owns the all-one codeword which means that it is possible to consider the generator matrix ${{\mathbf{G}}_{k\times N}}$ has a basis, say ${{\mathbf{g}}_{0}}$, whose all elements are one. Let the all-one column of ${{\mathbf{A}}_{N\times {{p}^{k}}}}$ be
\begin{equation}\label{Eqn:All_one_A}
\mathbf{a}_{N\times 1}^{0{{\alpha }^{(k-1)}}+0{{\alpha }^{(k-2)}}+\cdots +0{{\alpha }^{1}}+1}={{\mathbf{g}}_{0}}
\end{equation}
Considering the fact that the elements of ${{\mathbf{P}}_{{N}'\times {{p}^{k{k}'}}}}$ belong to the finite field ${{\mathbb{F}}_{{{p}^{k}}}}$,  it can be shown that ${{\mathbf{P}}_{{N}'\times {{p}^{k{k}'}}}}$ has the all-one codeword, i.e.,
\begin{equation}\label{Eqn:All_one_P}
\resizebox{0.9\hsize}{!}{$
\mathbf{p}_{N' \times 1}^1  = \left[ {\begin{array}{*{20}c}
   {0\alpha ^{(k - 1)}  +  \cdots  + 0\alpha ^1  + 1}  \\
    \vdots   \\
   {0\alpha ^{(k - 1)}  +  \cdots  + 0\alpha ^1  + 1}  \\

 \end{array} } \right]_{N' \times 1}  = \left[ {\begin{array}{*{20}c}
   1  \\
    \vdots   \\
   1  \\

 \end{array} } \right]$}
\end{equation}
Using (\ref{Eqn:All_one_A}), (\ref{Eqn:All_one_P}), (\ref{Eqn:CR_Deduction}), and the correspondence between the column vector $\mathbf{a}_{N\times 1}^{0{{\alpha }^{(k-1)}}+\cdots +0{{\alpha }^{1}}+1}$ and the element $0{{\alpha }^{(k-1)}}+\cdots +0{{\alpha }^{1}}+1$, one can immediately infer that ${{\mathbf{C}}_{N{N}'\times {{p}^{k{k}'}}}}$ owns an all-one column as
\begin{equation}\label{Eqn:All_one_C}
\text{B}\left( {\mathbf{p}_{N' \times 1}^1 } \right) = \left[ {\begin{array}{*{20}c}
   1 &  \cdots  & 1  \\
 \end{array} } \right]_{NN' \times 1}^T
\end{equation}
\end {proof}
This completes the proof.

In the following examples, using ${{\mathbf{C}}_{N{N}'\times {{p}^{k{k}'}}}}$, we apply Theorem  \ref{CS_Generation} to construct low coherence deterministic sensing matrices.
\subsection {Case Study} \label {Case_Study}
To employ Theorems \ref{CS_Generation} and \ref{Our_CR_Method}, we introduce some sensing matrices with low coherences and calculate their exact coherence value.

\begin{example} \label{49_343}
Let the generator matrix of codewords of ${{\mathbf{A}}_{N\times {{p}^{k}}}}$ and ${{\mathbf{P}}_{{N}'\times {{p}^{k{k}'}}}}$be ${{\mathbf{G}}_{2\times p}}={{\left[ {{\mathbf{g}}_{0}},{{\mathbf{g}}_{1}} \right]}^{T}}$, where ${{\mathbf{g}}_{0}}=\left[ 1,\cdots ,1 \right]_{1\times p}^{T}$ and ${{\mathbf{g}}_{1}}=\left[ 0,\cdots ,p-1 \right]_{1\times p}^{T}$. Moreover, suppose that the columns of ${{\mathbf{A}}_{N\times {{p}^{k}}}}$ and ${{\mathbf{P}}_{{N}'\times {{p}^{k{k}'}}}}$ are of the form ${{\gamma }_{0}}{{\mathbf{g}}_{0}}+{{\gamma }_{1}}{{\mathbf{g}}_{1}}$ and ${{\kappa }_{0}}{{\mathbf{g}}_{0}}+{{\kappa }_{1}}{{\mathbf{g}}_{1}}$, respectively, where ${{\gamma }_{0}},{{\gamma }_{1}}\in {{\mathbb{F}}_{p}}$ and ${{\kappa }_{0}},{{\kappa }_{1}}\in {{\mathbb{F}}_{{{p}^{2}}}}$. It is obvious that ${{\mathbf{A}}_{N\times {{p}^{k}}}}$, ${{\mathbf{P}}_{{N}'\times {{p}^{k{k}'}}}}$, and ${{\mathbf{C}}_{N{N}'\times {{p}^{k{k}'}}}}$ are of sizes $p\times {{p}^{2}}$, $p\times {{p}^{4}}$, and ${{p}^{2}}\times {{p}^{4}}$, respectively. Trivially, the minimum distances of linear codes created by ${{\mathbf{G}}_{2\times p}}$ is $p-1$ for both codewords of ${{\mathbf{A}}_{p\times {{p}^{2}}}}$ and ${{\mathbf{P}}_{p\times {{p}^{4}}}}$; i.e. ${{d}_{\min }}=p-1$ and ${{{d}'}_{\min }}=p-1$. The minimum distance of ${{\mathbf{C}}_{{{p}^{2}}\times {{p}^{4}}}}$ is thus equal to $d_{\min }^{\mathbf{C}}={{\left( p-1 \right)}^{2}}$ which is a large minimum distance. Let us apply Theorem \ref{CS_Generation} to this matrix. We have the following equation:

\begin{equation}\label{Eqn:Th1_Transform}
{{\mathbf{{C}''}}_{{{p}^{2}}\times {{p}^{3}}}}\text{=}\frac{1}{\sqrt{{{p}^{2}}}}{{\left[ {{e}^{j\frac{2\pi }{p}{{{{c}'}}_{rt}}}} \right]}_{r,t}}
\end{equation}
where ${{{c}'}_{rt}}$ is an element of ${{\mathbf{{C}'}}_{{{p}^{2}}\times {{p}^{3}}}}$ and this matrix is constructed from matrix ${{\mathbf{C}}_{{{p}^{2}}\times {{p}^{4}}}}$ in such a way that from each column set $\left\{ {{\mathbf{c}}_{{{p}^{2}}\times 1}},{{\mathbf{c}}_{{{p}^{2}}\times 1}}\oplus {{\mathbf{1}}_{{{p}^{2}}\times 1}},...,{{\mathbf{c}}_{{{p}^{2}}\times 1}}\oplus {{(\mathbf{p-1})}_{{{p}^{2}}\times 1}} \right\}$  only one column is selected.  We allege that the coherence of ${{\mathbf{{C}''}}_{{{p}^{2}}\times {{p}^{3}}}}$ is ${1}/{p}\;$.
\end{example}

\begin {proof}\renewcommand{\qedsymbol}{}
First, we can write,
\begin{equation}\label{Eqn:---}
{{\mathbf{C}}_{{{p}^{2}}\times {{p}^{4}}}}=\text{B}\left( {{\mathbf{P}}_{p\times {{p}^{4}}}} \right)
\end{equation}

As illustrated in (\ref{Eqn:All_one_C}), the all-one column of ${{\mathbf{C}}_{{{p}^{2}}\times {{p}^{4}}}}$ is obtained when the column replacement acts on the all-one column of ${{\mathbf{P}}_{p\times {{p}^{4}}}}$ as
\begin{equation}\label{Eqn:---}
\text{B}\left( {\mathbf{p}_{p \times 1}^1 } \right) = \left[ {\begin{array}{*{20}c}
   1 &  \cdots  & 1  \\

 \end{array} } \right]_{p^2  \times 1}^T  = \mathbf{1}_{p^2  \times 1}
\end{equation}
Imagine the column ${{\mathbf{c}}_{{{p}^{2}}\times 1}}$ of ${{\mathbf{C}}_{{{p}^{2}}\times {{p}^{4}}}}$ has the following form:
\begin{equation}\label{Eqn:---}
{{\mathbf{c}}_{{{p}^{2}}\times 1}}=\text{B}\left( \left( {{\lambda }_{11}}\alpha +{{\lambda }_{10}} \right){{\mathbf{g}}_{1}}+\left( {{\lambda }_{01}}\alpha +{{\lambda }_{00}} \right){{\mathbf{g}}_{0}} \right)
\end{equation}

where ${{\lambda }_{11}},{{\lambda }_{10}}\ ,{{\lambda }_{01}},{{\lambda }_{00}}\in {{\mathbb{F}}_{p}}$.
Using the definition of ${{\mathbf{C}}_{{{p}^{2}}\times {{p}^{4}}}}$, ${{\mathbf{c}}_{{{p}^{2}}\times 1}}$, and Theorem \ref{Our_CR_Method}, we can write
\begin{equation}\label{Eqn:---}
\resizebox{0.9\hsize}{!}{$
\begin{gathered}
  \mathbf{c}_{p^2  \times 1}  + \beta _0 (\mathbf{1})_{p^2  \times 1} \text{ = B}\left( {\mathop {\left( {\lambda _{11} \alpha  + \lambda _{10} } \right)\mathbf{g}_1  + }\limits_{\mathop {}\limits_{} } } \right. \hfill \\
  \;\;\;\;\;\;\;\;\;\;\;\;\;\;\;\;\;\;\;\;\;\;\;\;\;\;\;\;\;\;\;\;\;\;\;\;\; \left. {\left( {\lambda _{01} \alpha  + \underbrace {\lambda _{00}  + \beta _0 }_{\lambda '_{00} }} \right)\mathbf{g}_0 } \right) \hfill \\
\end{gathered}$}
\end{equation}

Obviously to choose one column of the column set \\ $\left\{ {{\mathbf{c}}_{{{p}^{2}}\times 1}},{{\mathbf{c}}_{{{p}^{2}}\times 1}}\oplus {{\mathbf{1}}_{{{p}^{2}}\times 1}},...,{{\mathbf{c}}_{{{p}^{2}}\times 1}}\oplus {{(\mathbf{p-1})}_{{{p}^{2}}\times 1}} \right\}$, it is enough to keep ${{\lambda }_{00}}$ constant while applying the column replacement over ${{\mathbf{P}}_{p\times {{p}^{4}}}}$. For convenience, we consider ${{\lambda }_{00}}=0$. Then, the relationship between ${{\mathbf{{C}'}}_{{{p}^{2}}\times {{p}^{3}}}}$ and ${{\mathbf{P}}_{p\times {{p}^{4}}}}$ is given by

\begin{equation}\label{Eqn:---}
{{\mathbf{{C}'}}_{{{p}^{2}}\times {{p}^{3}}}}=\text{B}\left( {{{\mathbf{{P}'}}}_{p\times {{p}^{3}}}} \right)
\end{equation}

where ${{\mathbf{{P}'}}_{p\times {{p}^{3}}}}$ contains those columns of ${{\mathbf{P}}_{p\times {{p}^{4}}}}$ which correspond to$\left( {{\lambda }_{11}}\alpha +{{\lambda }_{10}} \right){{\mathbf{g}}_{1}}+\left( {{\lambda }_{01}}\alpha  \right){{\mathbf{g}}_{0}}$. So far, we could show how it is feasible to generate ${{\mathbf{{C}'}_{{{p}^{2}}\times {{p}^{3}}}}}$ based on ${{\mathbf{P}}_{p\times {{p}^{4}}}}$. Let us proceed to construct the sensing matrix ${{\mathbf{{C}''}}_{{{p}^{2}}\times {{p}^{3}}}}$ based on (\ref{Eqn:Th1_Transform}) as
\begin{equation}\label{Eqn:---}
{ {{\left[ f\left( {{{{c}'}}_{rt}} \right) \right]}_{r,t}}=\left[ \begin{matrix}
   {{{\mathbf{{c}''}}}_{{{p}^{2}}\times 1}^{1}} & \begin{matrix}
   \cdots  & {{{\mathbf{{c}''}}}^{k}_{{{p}^{2}}\times 1}} & \cdots   \\
\end{matrix} & {{{\mathbf{{c}''}}}^{{{p}^{3}}}_{{{p}^{2}}\times 1}}  \\
\end{matrix} \right]}
\end{equation}

where $f\left( {{{{c}'}}_{rt}} \right)=\frac{1}{\sqrt{{{p}^{2}}}}{{e}^{j\frac{2\pi }{p}{{{{c}'}}_{rt}}}}$ is defined for simplicity. Since the elements of ${{\mathbf{{C}'}}_{{{p}^{2}}\times {{p}^{3}}}}$ come from the primary matrix ${{\mathbf{A}}_{p\times {{p}^{2}}}}$, i.e. ${{{c}'}_{rt}}\in {{\mathbf{A}}_{p\times {{p}^{2}}}}$, we can firstly apply the transformation function $f$ on each element of ${{\mathbf{A}}_{p\times {{p}^{2}}}}$ which leads to

\begin{equation}\label{Eqn:---}
\resizebox{0.9\hsize}{!}{$\mathbf{\bar A}_{p \times p^2 }  = \left[ {\mathbf{\bar a}_{p \times 1}^{0\alpha  + 0} , \cdots ,\mathbf{\bar a}_{p \times 1}^{\gamma _1 \alpha  + \gamma _0 } , \cdots ,\mathbf{\bar a}_{p \times 1}^{\left( {p - 1} \right)\alpha  + \left( {p - 1} \right)} } \right]$}
\end{equation}

where $\mathbf{\bar{a}}_{p\times 1}^{{{\gamma }_{1}}\alpha +{{\gamma }_{0}}}$ is any arbitrary column of ${{\mathbf{\bar{A}}}_{p\times {{p}^{2}}}}$ which corresponds to the column $\mathbf{a}_{p \times 1}^{\gamma _1 \alpha  + \gamma _0 }$ of matrix $\mathbf{A}_{p \times p^2 }$, ${{\gamma }_{0}},{{\gamma }_{1}}\in {{\mathbb{F}}_{p}}$ and  superscript ${{\gamma }_{1}}\alpha +{{\gamma }_{0}}$ is considered to accentuate the correspondence of columns of ${{\mathbf{\bar{A}}}_{p\times {{p}^{2}}}}$ and the elements of ${{\mathbf{{P}'}}_{p\times {{p}^{3}}}}$.
To simplify the notations, we have considered the defined transformation function $f$ on matrices and vectors as
\begin{equation}\label{Eqn:---}
\begin{gathered}
  \mathbf{\bar A}_{p \times p^2 }  = f\left( {\mathbf{A}_{p \times p^2 } } \right) \hfill \\
  \mathbf{\bar a}_{p \times 1}^{\gamma _1 \alpha  + \gamma _0 }  = f\left( {\mathbf{a}_{p \times 1}^{\gamma _1 \alpha  + \gamma _0 } } \right) \hfill \\
\end{gathered}
\end{equation}
Then, we can rewrite ${{\mathbf{{C}''}}_{{{p}^{2}}\times {{p}^{3}}}}$ as
\begin{equation}\label{Eqn:---}
\mathbf{C''}_{p^2  \times p^3 }  = {\rm B'}\left( {\mathbf{P'}_{p \times p^3 } } \right)
\end{equation}

where ${\rm B'}$ represents the column replacement for ${{\mathbf{\bar{A}}}_{p\times {{p}^{2}}}}$. Now, we can calculate the coherence of ${{\mathbf{{C}''}_{{{p}^{2}}\times {{p}^{3}}}}}$ whose coherence is
\begin{equation}\label{Eqn:Coh_EX1}{{\mu }_{{\mathbf{{C}''}}}}=\underset{j\ne k}{\mathop{\max }}\,\left| \left\langle {{{\mathbf{{c}''}}}^{j}_{{{p}^{2}}\times 1}},{{{\mathbf{{c}''}}}^{k}_{{{p}^{2}}\times 1}} \right\rangle  \right|\end{equation}

where$\left\langle \centerdot ,\centerdot  \right\rangle $ illustrates the inner product of two vectors and ${{\mathbf{{c}''}}^{j}_{{{p}^{2}}\times 1}}, {{\mathbf{{c}''}}^{k}_{{{p}^{2}}\times 1}}$ are two different columns of ${{\mathbf{{C}''}}_{{{p}^{2}}\times {{p}^{3}}}}$. Considering  ${\bf c''}^j _{p^2  \times 1} ,\;{\bf c''}^k _{p^2  \times 1} $  as
\begin{equation}\label{Eqn:---}
\resizebox{0.98\hsize}{!}{$
\begin{array}{l}
 {\bf c''}^j _{p^2  \times 1}  = {\rm B'}\left( {{\bf p'}^j _{p \times 1} } \right) = {\rm B'}\left( {\left( {\lambda _{11} \alpha  + \lambda _{10} } \right){\bf g}_1  + \left( {\lambda _{01} \alpha } \right){\bf g}_0 } \right) \\
 {\bf c''}^k _{p^2  \times 1}  = {\rm B'}\left( {{\bf p'}^k _{p \times 1} } \right) = {\rm B'}\left( {\left( {\lambda '_{11} \alpha  + \lambda '_{10} } \right){\bf g}_1  + \left( {\lambda '_{01} \alpha } \right){\bf g}_0 } \right) \\
 \end{array}$}
\end{equation}
where ${{\mathbf{{p}'}}^{j}_{p\times 1}}$ and ${{\mathbf{{p}'}}^{k}_{p\times 1}}$ are the ${{j}^{th}}$ and ${{k}^{th}}$ columns of ${{\mathbf{{P}'}}_{p\times {{p}^{3}}}}$,  respectively.
By replacing the corresponding columns of ${{\mathbf{\bar{A}}}_{p\times {{p}^{2}}}}$, (\ref{Eqn:Coh_EX1}) reduces to
\begin{equation}\label{Eqn:Coh_EX1_1}
{
\begin{gathered}
  \hspace{-6.5cm}\mu _{\mathbf{C''}}  = \\ \mathop {\max }\limits_{\begin{array}{*{20}{c}}
  {\left\{ \begin{subarray}{l} 
  {\lambda _{11}},{\lambda _{10}}\;, \\ 
  {\lambda _{01}},{{\lambda '}_{11}}, \\ 
  {{\lambda '}_{10}}\;,{{\lambda '}_{01}} 
\end{subarray}  \right.}&  \hspace{-0.35cm}{ \in {\mathbb{F}_p}} 
\end{array}} \left| {\left\langle {\left[ {\begin{array}{*{20}c}
   {\mathbf{\bar a}_{p \times 1}^{\left( {\lambda _{11} \alpha  + \lambda _{10} } \right)0 + \lambda _{01} \alpha } }  \\
   {\begin{array}{*{20}c}
    \vdots   \\
   {\mathbf{\bar a}_{p \times 1}^{\left( {\lambda _{11} \alpha  + \lambda _{10} } \right)\phi  + \left( {\lambda _{01} \alpha } \right)} }  \\

 \end{array} }  \\
    \vdots   \\
   {\mathbf{\bar a}_{p \times 1}^{\left( {\lambda _{11} \alpha  + \lambda _{10} } \right)\left( {p - 1} \right) + \left( {\lambda _{01} \alpha } \right)} }  \\

 \end{array} } \right],} \right.} \right. \hfill \\
  \;\;\;\;\;\;\;\;\;\;\;\;\;\;\;\;\;\;\;\;\;\;\;\;\;\;\;\;\;\;\;\;\;\;\;\;\;\;\;\;\;\;\;\;\;\;\;\left.   \hspace{-1.5cm}{\left. {\left[ {\begin{array}{*{20}c}
   {\mathbf{\bar a}_{p \times 1}^{\left( {\lambda '_{11} \alpha  + \lambda '_{10} } \right)0 + \lambda '_{01} \alpha } }  \\
   {\begin{array}{*{20}c}
    \vdots   \\
   {\mathbf{\bar a}_{p \times 1}^{\left( {\lambda '_{11} \alpha  + \lambda '_{10} } \right)\phi  + \left( {\lambda '_{01} \alpha } \right)} }  \\

 \end{array} }  \\
    \vdots   \\
   {\mathbf{\bar a}_{p \times 1}^{\left( {\lambda '_{11} \alpha  + \lambda '_{10} } \right)\left( {p - 1} \right) + \left( {\lambda '_{01} \alpha } \right)} }  \\

 \end{array} } \right]} \right\rangle } \right| \hfill \\
\end{gathered} }
\end{equation}

where $\phi \in {{\mathbb{F}}_{p}}$.  Using  $f$, (\ref{Eqn:Coh_EX1_1}) becomes
\begin{equation}\label{Eqn:Eqn:Coh_EX1_1_1}
{
\begin{gathered}
  \hspace{-6.5cm}\mu _{\mathbf{C''}}  = \\ \mathop {\max }\limits_{\begin{array}{*{20}{c}}
  {\left\{ \begin{subarray}{l} 
  {\lambda _{11}},{\lambda _{10}}\;, \\ 
  {\lambda _{01}},{{\lambda '_{11}}}, \\ 
  {{\lambda '_{10}}}\;,{{\lambda '_{01}}} 
\end{subarray}  \right.}&  \hspace{-0.35cm}{ \in {\mathbb{F}_p}} \end{array}} \left| {\left\langle {\left[ {\begin{array}{*{20}c}
   {f\left( {\mathbf{a}_{p \times 1}^{\left( {\lambda _{11} \alpha  + \lambda _{10} } \right)0 + \lambda _{01} \alpha } } \right)}  \\
   {\begin{array}{*{20}c}
    \vdots   \\
   {f\left( {\mathbf{a}_{p \times 1}^{\left( {\lambda _{11} \alpha  + \lambda _{10} } \right)\phi  + \left( {\lambda _{01} \alpha } \right)} } \right)}  \\

 \end{array} }  \\
    \vdots   \\
   {f\left( {\mathbf{a}_{p \times 1}^{\left( {\lambda _{11} \alpha  + \lambda _{10} } \right)\left( {p - 1} \right) + \left( {\lambda _{01} \alpha } \right)} } \right)}  \\

 \end{array} } \right],} \right.} \right. \hfill \\
  \;\;\;\;\;\;\;\;\;\;\;\;\;\;\;\;\;\;\;\;\;\;\;\;\;\;\;\;\;\;\;\;\;\;\;\;\;\;\;\;\;\;\;\;\;\;\;\;\left. \hspace{-1.7cm}{\left. {\left[ {\begin{array}{*{20}c}
   {f\left( {\mathbf{a}_{p \times 1}^{\left( {\lambda '_{11} \alpha  + \lambda '_{10} } \right)0 + \lambda '_{01} \alpha } } \right)}  \\
   {\begin{array}{*{20}c}
    \vdots   \\
   {f\left( {\mathbf{a}_{p \times 1}^{\left( {\lambda '_{11} \alpha  + \lambda '_{10} } \right)\phi  + \left( {\lambda '_{01} \alpha } \right)} } \right)}  \\

 \end{array} }  \\
    \vdots   \\
   {f\left( {\mathbf{a}_{p \times 1}^{\left( {\lambda '_{11} \alpha  + \lambda '_{10} } \right)\left( {p - 1} \right) + \left( {\lambda '_{01} \alpha } \right)} } \right)}  \\

 \end{array} } \right]} \right\rangle } \right| \hfill \\
\end{gathered} }
\end{equation}

Contemplate the following characteristic of $f$ over two either scalars or vectors ${{\mathbf{u}}_{1}}$ and  ${{\mathbf{u}}_{2}}$ as
\begin{equation}\label{Eqn:---}
f\left( {\mathbf{u}_1 } \right).f^* \left( {\mathbf{u}_2 } \right) = f\left( {\mathbf{u}_1  - \mathbf{u}_2 } \right)
\end{equation}
where $*$ indicates complex conjugation. Then, the inner part of  (\ref{Eqn:Eqn:Coh_EX1_1_1}) can be simplified as
\begin{equation}\label{Eqn:Coh_EX1_2}
\resizebox {1. \hsize} {!} { $
\hspace{-0.2cm}\left\langle \hspace{-0.09cm}{\left[ \hspace{-0.2cm}{\begin{array}{*{20}c}
   {f\left( {\mathbf{a}_{p \times 1}^{\left( {\lambda _{11} \alpha  + \lambda _{10} } \right)0 + \lambda _{01} \alpha }  - \mathbf{a}_{p \times 1}^{\left( {\lambda '_{11} \alpha  + \lambda '_{10} } \right)0 + \lambda '_{01} \alpha } } \right)}  \\
   {\begin{array}{*{20}c}
    \vdots   \\
   {f\left( {\mathbf{a}_{p \times 1}^{\left( {\lambda _{11} \alpha  + \lambda _{10} } \right)\phi  + \left( {\lambda _{01} \alpha } \right)}  - \mathbf{a}_{p \times 1}^{\left( {\lambda '_{11} \alpha  + \lambda '_{10} } \right)\phi  + \left( {\lambda '_{01} \alpha } \right)} } \right)}  \\

 \end{array} }  \\
    \vdots   \\
   {f\left( {\mathbf{a}_{p \times 1}^{\left( {\lambda _{11} \alpha  + \lambda _{10} } \right)\left( {p - 1} \right) + \left( {\lambda _{01} \alpha } \right)}  - \mathbf{a}_{p \times 1}^{\left( {\lambda '_{11} \alpha  + \lambda '_{10} } \right)\left( {p - 1} \right) + \left( {\lambda '_{01} \alpha } \right)} } \right)}  \\

 \end{array} } \hspace{-0.2cm} \right]\hspace{-0.13cm},\mathbf{1}_{p^2  \times 1} }\hspace{-0.12cm} \right\rangle $}
\end{equation}

Moreover, according to (\ref{Eqn:Primary_Mat}), the inner term of transformation function $f$ in (\ref{Eqn:Coh_EX1_2}) for an arbitrary $\phi \in {{\mathbb{F}}_{p}}$ can be rewritten as
\begin{equation}\label{Eqn:---}
\begin{gathered}
  \mathbf{a}_{p \times 1}^{\left( {\lambda _{11} \alpha  + \lambda _{10} } \right)\phi  + \left( {\lambda _{01} \alpha } \right)}  - \mathbf{a}_{p \times 1}^{\left( {\lambda '_{11} \alpha  + \lambda '_{10} } \right)\phi  + \left( {\lambda '_{01} \alpha } \right)}  =   \\
  \mathbf{a}_{p \times 1}^{\left( {\lambda _{11} \alpha  + \lambda _{10} } \right)\phi  + \left( {\lambda _{01} \alpha } \right) - \left( {\left( {\lambda '_{11} \alpha  + \lambda '_{10} } \right)\phi  + \left( {\lambda '_{01} \alpha } \right)} \right)} \;\; =  \hfill \\
  \mathbf{a}_{p \times 1}^{\left( {\lambda ''_{11} \alpha  + \lambda ''_{10} } \right)\phi  + \left( {\lambda ''_{01} \alpha } \right)}  \hfill \\
\end{gathered}
\end{equation}

where ${{{\lambda }''_{11}}}={{{\lambda }'}_{11}}-{{\lambda }_{11}},{{{\lambda }''_{10}}}={{{\lambda }'_{10}}}-{{\lambda }_{10}},{{{\lambda }''_{01}}}={{{\lambda }'_{01}}}-{{\lambda }_{01}}$and ${{{\lambda }''_{11}}},{{{\lambda }''_{10}}},{{{\lambda }''_{01}}}\in {{\mathbb{F}}_{p}}$. Therefore, (\ref{Eqn:Coh_EX1_2}) can be stated as

\begin{equation}\label{Eqn:---}
\resizebox{1\hsize}{!}{$
{
\begin{gathered}
  {\mu _{{\mathbf{C''}}}} =  \hfill \\
 \hspace{-0.4cm} \mathop {\max }\limits_{\begin{array}{*{20}{c}}
  {\left\{ \begin{subarray}{l} 
  {{\lambda ''_{01}}}, \\ 
  {{\lambda ''_{10}}}, \\ 
  {{\lambda ''_{11}}} 
\end{subarray}  \right.}&\hspace{-0.3cm}{ \in {\mathbb{F}_p}} 
\end{array}} \hspace{-0.2cm}\left| {\left\langle {\left[\hspace{-0.15cm} {\begin{array}{*{20}{c}}
  {f\left( {{{\lambda ''_{01}}}{{\mathbf{g}}_1}} \right)} \\ 
  {\begin{array}{*{20}{c}}
   \vdots  \\ 
  {f\left( {\left( {{{\lambda ''_{11}}}\phi  + {{\lambda ''_{01}}}} \right){{\mathbf{g}}_1} + {{\lambda ''_{10}}}{{\mathbf{g}}_0}} \right)} 
\end{array}} \\ 
   \vdots  \\ 
  {f\left( {\left( {{{\lambda ''_{11}}}\left( {p - 1} \right) + {{\lambda ''_{01}}}} \right){{\mathbf{g}}_1} + {{\lambda ''_{10}}}{{\mathbf{g}}_0}} \right)} 
\end{array}} \hspace{-0.17cm}\right]\hspace{-0.1cm},{{\mathbf{1}}_{{p^2} \times 1}}} \right\rangle } \right| \hfill \\ 
\end{gathered}
}$}
\end{equation}
Owing to the fact that coherence is considered over two different columns, ${{{\lambda }''_{01}}},{{{\lambda }''_{10}}},{{{\lambda }''_{11}}}$ cannot be zero at the same time. Considering ${{\mathbf{g}}_{1}}={{\left[ 0,\cdots ,p-1 \right]}^{T}}$ and ${{\mathbf{g}}_{0}}=\left[ 1,\cdots ,1 \right]_{1\times p}^{T}$, ${{\mu }_{{\mathbf{{C}''}}}}$ is written as
\end{proof}


\begin{align} \label{Eqn:Shrinked_equation_Coh}
\begin{split}
\resizebox{1\hsize}{!}{$
\begin{gathered}
  {\mu _{{\mathbf{C''}}}} =  \hfill \\
  \hspace{-0.2cm}\mathop {\max }\limits_{\left\{ {\begin{array}{*{20}{c}}
\hspace{-0.2cm}  \begin{subarray}{l} 
  {{\lambda ''_{01}}}, \\ 
  {{\lambda ''_{10}}}, \\ 
  {{\lambda ''_{11}}} 
\end{subarray} &\hspace{-0.3cm}{ \in {\mathbb{F}_p}} 
\end{array}} \right.} \hspace{-0.25cm}\frac{1}{{{p^2}}}\left| {\underbrace {f\left( {0{{\lambda ''_{01}}}} \right) +  \cdots  + f\left( {\gamma {{\lambda ''_{01}}}} \right) +  \cdots  + f\left( {\left( {p - 1} \right){{\lambda ''_{01}}}} \right)}_{\phi  = 0} +  \cdots } \right. +  \hfill \\
  \,\,\,\,\,\,\,\,\,\,\,\,\,\,\,\,\,\,\,\underbrace {\begin{array}{*{20}{c}}
  {f\left( {{{\lambda ''_{10}}}\phi } \right) +  \cdots  + f\left( {\left( {{{\lambda ''_{11}}}\phi  + {{\lambda ''_{01}}}} \right)\gamma  + {{\lambda ''_{10}}}\phi } \right) +  \cdots  + } \\ 
  {f\left( {\left( {{{\lambda ''_{11}}}\phi  + {{\lambda ''_{01}}}} \right)\left( {p - 1} \right) + {{\lambda ''_{10}}}\phi } \right)} 
\end{array}}_\phi \;\; +  \cdots  +  \hfill \\
  \,\,\,\,\,\,\,\,\,\,\left. {\underbrace {\begin{array}{*{20}{c}}
  {\;f\left( {{{\lambda ''_{10}}}\left( {p - 1} \right)} \right) +  \cdots  + f\left( {\left( {{{\lambda ''_{11}}}\left( {p - 1} \right) + {{\lambda ''_{01}}}} \right)\gamma  + {{\lambda ''_{10}}}\left( {p - 1} \right)} \right)} \\ 
  { +  \cdots  + f\left( {\left( {{{\lambda ''_{11}}}\left( {p - 1} \right) + {{\lambda ''_{01}}}} \right)\left( {p - 1} \right) + {{\lambda ''_{10}}}\left( {p - 1} \right)} \right)} 
\end{array}}_{\phi  = p - 1}} \right| \hfill \\
  \,\,\,\,\,\,\,\, = \mathop {\max }\limits_{{{\lambda ''_{01}}},{{\lambda ''_{10}}},{{\lambda ''_{11}}} \in {\mathbb{F}_p}} \frac{1}{{{p^2}}}\left| {\sum\limits_{\phi  = 0}^{p - 1} {\sum\limits_{\gamma  = 0}^{p - 1} {{e^{j\frac{{2\pi }}{p}\left( {\gamma \left( {{{\lambda ''_{11}}}\phi  + {{\lambda ''_{01}}}} \right) + {{\lambda ''_{10}}}\phi } \right)}}} } } \right| \hfill \\ 
\end{gathered} $}
\end{split}
\end{align}

By introducing the additive character over a given finite field, defined in the following, ${{\mu }_{{\mathbf{{C}''}}}}$ is equal to ${1}/{p}\;$ as illustrated in Appendix A.

\begin{defi} \label{}
The canonical additive character of finite field ${{\mathbb{F}}_{p}}$ is (Ch. 5 \cite{Finite_22}):
\begin{equation}\label{Eqn:---}\chi \left( \gamma  \right)={{e}^{j\frac{2\pi }{p}\gamma }},\ \ \ \ \ \ \ \forall \gamma \in {{\mathbb{F}}_{p}}\end{equation}
where $p$ is a prime integer.  Then, the following equations hold for $\chi \left( \gamma  \right)$:

\begin{equation}\label{Eqn:---}\chi \left( {{\gamma }_{1}}+{{\gamma }_{2}} \right)=\chi \left( {{\gamma }_{1}} \right)\chi \left( {{\gamma }_{2}} \right)\ \ \ \ \ \ \ \forall {{\gamma }_{1}},{{\gamma }_{2}}\in {{\mathbb{F}}_{p}}\end{equation}
\begin{equation}\label{Eqn:Additive_Char}\sum\limits_{\gamma \in {{\mathbb{F}}_{p}}}{\chi \left( \tau *\gamma +\beta  \right)}=\left\{ \begin{matrix}
   0,\ \ \ \ \ \tau \in {{\mathbb{F}}_{p}}-\left\{ 0 \right\}\   \\
   p,\ \ \ \ \tau =0\ \ \ \ \ \ \ \ \ \   \\
\end{matrix} \right.\ \ \ \ \ \ \ \beta \in {{\mathbb{F}}_{p}}\end{equation}

\end {defi}

Within this example we could offer low coherence deterministic sensing matrices of size ${{p}^{2}}\times {{p}^{3}}$  and coherence ${1}/{p}\;$. Note that this resembles to the sensing matrices addressed in \cite{Mohades_15} showing the generality of Theorem \ref{Our_CR_Method}.

\begin {example} \label{42_343}
Considering the same procedure as Example \ref{49_343}, we only highlight the changes.
Let the generator matrix of codewords of ${{\mathbf{A}}_{N\times {{p}^{k}}}}$ and ${{\mathbf{P}}_{{N}'\times {{p}^{k{k}'}}}}$ be ${{\mathbf{G}}_{2\times p}}={{\left[ {{\mathbf{g}}_{0}},{{\mathbf{g}}_{1}} \right]}^{T}}$and ${{\mathbf{{G}'}}_{2\times p}}={{\left[ {{{\mathbf{{g}'}}}_{0}},{{{\mathbf{{g}'}}}_{1}} \right]}^{T}}$, respectively, where ${{\mathbf{g}}_{0}}=\left[ 1,\cdots ,1 \right]_{1\times p}^{T}$, ${{\mathbf{g}}_{1}}=\left[ 0,\cdots ,p-1 \right]_{1\times p}^{T}$, ${{\mathbf{{g}'}}_{0}}=\left[ 1,\cdots ,1 \right]_{1\times \left( p-1 \right)}^{T}$ and ${{\mathbf{{g}'}}_{1}}=\left[ 0,\cdots ,p-2 \right]_{1\times \left( p-1 \right)}^{T}$. Inspiring the previous example, we have ${{\mathbf{A}}_{p\times {{p}^{2}}}}$, ${{\mathbf{P}}_{\left( p-1 \right)\times {{p}^{4}}}}$, and ${{\mathbf{C}}_{p\left( p-1 \right)\times {{p}^{4}}}}$ for which ${{d}_{\min }}=p-1$, ${{{d}'}_{\min }}=p-2$, and $d_{\min }^{\mathbf{C}}=\left( p-1 \right)\left( p-2 \right)$, which is a large minimum distance. By applying Theorem \ref{CS_Generation} to ${{\mathbf{C}}_{p\left( p-1 \right)\times {{p}^{4}}}}$, we get
\begin{equation}\label{Eqn:---}{{\mathbf{{C}''}}_{p\left( p-1 \right)\times {{p}^{3}}}}\text{=}\frac{1}{\sqrt{p\left( p-1 \right)}}{{\left[ {{e}^{j\frac{2\pi }{p}{{{{c}'}}_{rt}}}} \right]}_{r,t}}\end{equation}

where ${{{c}'}_{rt}}$ is an element of ${{\mathbf{{C}'}}_{p\left( p-1 \right)\times {{p}^{3}}}}$ constructed from ${{\mathbf{C}}_{p\left( p-1 \right)\times {{p}^{4}}}}$ by selecting only one column from each column set $\left\{ {\mathbf{c}_{p\left( {p - 1} \right) \times 1} ,\mathbf{c}_{p\left( {p - 1} \right) \times 1}  \oplus \mathbf{1},...,\mathbf{c}_{p\left( {p - 1} \right) \times 1}  \oplus (\mathbf{p - 1})} \right\}$.  It is proved that the coherence of ${{\mathbf{{C}''}}_{p\left( p-1 \right)\times {{p}^{3}}}}$ is ${1}/{\left( p-1 \right)}\;$ (see Appendix B).
In this example, we proposed another set of low coherence deterministic sensing matrices of size $p\left( p-1 \right)\times {{p}^{3}}$  and coherence ${1}/{\left( p-1 \right)}\;$. Note that although similar matrices of size $\left( {{p}^{2}}-1 \right)\times {{p}^{3}}$ with the coherence of ${1}/{\left( p-1 \right)}\;$ are developed in \cite{Amini_19}, our method achieves the same coherence, but with a smaller number of rows in matrix construction which leads to a lower number of measurements.

\end{example}

\subsection {Matrix Resizing} \label{Matrix_Resizing}
Here, we employ the column replacement concept to resize any deterministic sensing matrix and we illustrate that after resizing to a matrix with a given greater size, our method leads to a matrix with lower coherence than that of \cite{Amini_19}. To do so, first we indicate Theorem \ref{Kronecker_Resizing_Method} from \cite{Amini_19} and then develop our method in Theorem \ref{CR_Resizing_Method}.

\begin{theorem} \label{Kronecker_Resizing_Method}
Contemplate two sensing matrices $\mathbf{A}$ and $\mathbf{B}$ whose coherence are ${{\mu }_{\mathbf{A}}}$and ${{\mu }_{\mathbf{B}}}$, respectively. Assuming $\mathbf{C=A}\otimes \mathbf{B}$, then the coherence of  $\mathbf{C}$ is ${{\mu }_{\mathbf{A}}}=\max \left\{ {{\mu }_{\mathbf{A}}},{{\mu }_{\mathbf{B}}} \right\}$ where $\otimes $ denotes the Kronecker product defined as \cite{Amini_19}
\begin{equation}\label{Eqn:---}
\mathbf{X = Z} \otimes \mathbf{Y} = \left[ {\begin{array}{*{20}c}
   {z_{11} \mathbf{Y}} &  \cdots  & {z_{1n} \mathbf{Y}}  \\
    \vdots  &  \ddots  &  \vdots   \\
   {z_{m1} \mathbf{Y}} &  \cdots  & {z_{mn} \mathbf{Y}}  \\
 \end{array} } \right]
\end{equation}

where  $\mathbf{Z}=\left[ \begin{matrix}
   {{z}_{11}} & \cdots  & {{z}_{1n}}  \\
   \vdots  & \ddots  & \vdots   \\
   {{z}_{m1}} & \cdots  & {{z}_{mn}}  \\
\end{matrix} \right]$ and $\mathbf{Y}$ are two arbitrary matrices.
\end{theorem}

Now we propose our method to resize deterministic sensing matrices based on column replacement and then by providing one example, we show outperformance of our method in comparison to that of \cite{Amini_19}.

\begin{theorem} \label{CR_Resizing_Method}
Let \scalebox{0.8}[1]{$\mathbf{A}_{m \times n}  = \left[ {\begin{array}{*{20}c}
   {\mathbf{a}_{m \times 1}^{t_1 } } &  \cdots  & {\mathbf{a}_{m \times 1}^{t_j } } &  \cdots  & {\mathbf{a}_{m \times 1}^{t_n } }  \\  \end{array} } \right]$} be a deterministic sensing matrix with coherence of ${{\mu }_{\mathbf{A}}}$, where superscript ${{t}_{j}}$ is to create a correspondence between the primary and pattern matrices in column replacement and $\mathbf{a}_{m\times 1}^{{{t}_{j}}}$ are normalized vectors for $1\le j\le n$. Furthermore, let \scalebox{0.8}[1]{$\mathbf{P}_{N \times L}  = \left[ {\begin{array}{*{20}c}
   {\mathbf{p}_{N \times 1}^1 } &  \cdots  & {\mathbf{p}_{N \times 1}^i } &  \cdots  & {\mathbf{p}_{N \times 1}^L }  \\
 \end{array} } \right]$} be a matrix whose elements are ${{t}_{j}},\ \ 1\le j\le n$ and no two distinct columns of this matrix have more than ${{d}_{\mathbf{P}}}$ elements in common. Trivially, $1\le i\le L$ and $0\le {{d}_{\mathbf{P}}}\le N-1$ are true statements. Then, the coherence of ${{\mathbf{C}}_{mN\times L}}\text{=}\frac{1}{\sqrt{N}}\text{B}\left( {{\mathbf{P}}_{N\times L}} \right)$ is ${{\mu }_{\mathbf{C}}}\le \frac{1}{N}\left( {{d}_{\mathbf{P}}}+\left( N-{{d}_{\mathbf{P}}} \right){{\mu }_{\mathbf{A}}} \right)$, where $\text{B}\left( {{\mathbf{P}}_{N\times L}} \right)$ means column replacement of ${{\mathbf{A}}_{m\times n}}$ based on the pattern matrix ${{\mathbf{P}}_{N\times L}}$ and the factor $\frac{1}{\sqrt{N}}$  normalizes the columns after column replacement.
\end{theorem}

\begin {proof}\renewcommand{\qedsymbol}{}
Suppose that \scalebox{.85}[1]{$\mathbf{C}_{mN \times L}  = \left[ {\mathbf{c}_{mN \times 1}^1  \cdots \mathbf{c}_{mN \times 1}^i  \cdots \mathbf{c}_{mN \times 1}^L } \right]$}. Then, ${{\mu }_{\mathbf{C}}}$ is defined as
\begin{equation}\label{Eqn:---}{{\mu }_{\mathbf{C}}}=\underset{1\le {{i}_{1}},{{i}_{2}}\le L}{\mathop{\max }}\,\left| \left\langle \mathbf{c}_{mN\times 1}^{{{i}_{1}}},\mathbf{c}_{mN\times 1}^{{{i}_{2}}} \right\rangle  \right|\end{equation}

Assume that the two arbitrary columns $\mathbf{c}_{mN\times 1}^{{{i}_{1}}}$ and $\mathbf{c}_{mN\times 1}^{{{i}_{2}}}$ are constructed by employing column replacement over two columns $\mathbf{p}_{N\times 1}^{{{i}_{1}}}$ and $\mathbf{p}_{N\times 1}^{{{i}_{2}}}$ of matrix ${{\mathbf{P}}_{N\times L}}$ as
 \begin{equation}\label{Eqn:---}
\scalebox{.85}[1]{$
\begin{gathered}
  \mathbf{c}_{mN \times 1}^{i_1 }  = \frac{1}
{{\sqrt N }}\text{B}\left( {\mathbf{p}_{N \times 1}^{i_1 } } \right) \hfill \\
  \;\;\;\;\;\;\;\; = \;\;\frac{1}
{{\sqrt N }}\left[ {\left[ {\text{B}\left( {p_1^{i_1 } } \right)} \right]_{1 \times m}^T  \cdots \left[ {\text{B}\left( {p_s^{i_1 } } \right)} \right]_{1 \times m}^T  \cdots \left[ {\text{B}\left( {p_N^{i_1 } } \right)} \right]_{1 \times m}^T } \right]^T  \hfill \\
\end{gathered}$}
\end{equation}

\begin{equation}\label{Eqn:---}
\scalebox{.85}[1]{$
\begin{gathered}
  \mathbf{c}_{mN \times 1}^{i_2 }  = \frac{1}
{{\sqrt N }}\text{B}\left( {\mathbf{p}_{N \times 1}^{i_2 } } \right) \hfill \\
  \;\;\;\;\;\;\;\; = \;\;\frac{1}
{{\sqrt N }}\left[ {\left[ {\text{B}\left( {p_1^{i_2 } } \right)} \right]_{1 \times m}^T  \cdots \left[ {\text{B}\left( {p_s^{i_2 } } \right)} \right]_{1 \times m}^T  \cdots \left[ {\text{B}\left( {p_N^{i_2 } } \right)} \right]_{1 \times m}^T } \right]^T  \hfill \\
\end{gathered} $}
\end{equation}
\end{proof}

where $1\le s\le N$and $1\le {{i}_{1}},{{i}_{2}}\le L$. Then, we have

\begin{align}\label{Eqn:Coh_CR}
\begin{gathered}
  \hspace{-5.4cm}\left| {\left\langle {\mathbf{c}_{mN \times 1}^{i_1 } ,\mathbf{c}_{mN \times 1}^{i_2 } } \right\rangle } \right| = \\ \frac{1}
{N}\left| {\left\langle {\left[ {\left( {\text{B}\left( {p_1^{i_1 } } \right)} \right)_{1 \times m}^T  \cdots \left( {\text{B}\left( {p_s^{i_1 } } \right)} \right)_{1 \times m}^T  \cdots \left( {\text{B}\left( {p_N^{i_1 } } \right)} \right)_{1 \times m}^T } \right]^T \hspace{-.2cm},} \right.} \right. \hfill \\ \hspace{-2.9cm}
  \;\;\;\;\;\;\;\;\;\;\;\;\;\;\;\;\;\;\;\;\;\;\;\;\;\;\;\;\;\;\;\;\;\;\;\left. {\left. {\left[ {\left( {\text{B}\left( {p_1^{i_2 } } \right)} \right)_{1 \times m}^T  \cdots \left( {\text{B}\left( {p_s^{i_2 } } \right)} \right)_{1 \times m}^T  \cdots \left( {\text{B}\left( {p_N^{i_2 } } \right)} \right)_{1 \times m}^T } \right]^T } \right\rangle } \right| \hfill \\
\end{gathered} \end{align}

According to the triangle inequality,  (\ref{Eqn:Coh_CR})  can be rewritten as an inequality as

\begin{align}\label{Eqn:ineqality_CR}
\begin{gathered}
 \hspace{-5cm} \left| {\left\langle {\mathbf{c}_{mN \times 1}^{i_1 } ,\mathbf{c}_{mN \times 1}^{i_2 } } \right\rangle } \right| \leqslant \\ \hspace{0cm}\;\;\;\frac{1}
{N}\left( {\left| {\left\langle {\left[ {\text{B}\left( {p_1^{i_1 } } \right)} \right]_{1 \times m}^T ,\left[ {\text{B}\left( {p_1^{i_2 } } \right)} \right]_{1 \times m}^T } \right\rangle } \right|} \right. +  \cdots  +  \\ \;\;\;\;\;\;\;\;\; \left| {\left\langle {\left[ {\text{B}\left( {p_s^{i_1 } } \right)} \right]_{1 \times m}^T ,\left[ {\text{B}\left( {p_s^{i_2 } } \right)} \right]_{1 \times m}^T } \right\rangle } \right| +  \cdots  +  \hfill \\
\;\;\;\;\;\;\;\;\;\;\left. {\left| {\left\langle {\left[ {\text{B}\left( {p_N^{i_1 } } \right)} \right]_{1 \times m}^T ,\left[ {\text{B}\left( {p_N^{i_2 } } \right)} \right]_{1 \times m}^T } \right\rangle } \right|} \right) \hfill \\
\end{gathered}
\end{align}

In Theorem \ref{CR_Resizing_Method}, we assumed that each two distinct columns of ${{\mathbf{P}}_{N\times L}}$ have at most ${{d}_{\mathbf{P}}}$ elements in common and considering the fact that $\text{B}\left( p_{s}^{{}} \right)$ recalls a column of ${{\mathbf{A}}_{m\times n}}$ whose coherence is ${{\mu }_{\mathbf{A}}}$, we have
\begin{equation}\label{Eqn:---}
\scalebox{.85}[1]{$
\left| {\left\langle {\left[ {\text{B}\left( {p_s^{i_1 } } \right)} \right]_{1 \times m}^T ,\left[ {\text{B}\left( {p_s^{i_2 } } \right)} \right]_{1 \times m}^T } \right\rangle } \right| = \left\{ {\begin{array}{*{20}c}
   {1\;\;\;\;\;\;\;\;\;\;p_s^{i_1 }  = p_s^{i_2 } }  \\
   {\;\mu  \leqslant \mu _\mathbf{A} \;\;\;p_s^{i_1 }  \ne p_s^{i_2 } }  \\
 \end{array} } \right.$}
\end{equation}

where $1 \leqslant s \leqslant N$. As mentioned, $p_{s}^{{{i}_{1}}}=p_{s}^{{{i}_{2}}}$ at most for ${{d}_{\mathbf{P}}}$ elements and the remainder are not equal. Therefore, the inequality (\ref{Eqn:ineqality_CR}) reduces to
\begin{equation}\label{Eqn:---}\left| \left\langle \mathbf{c}_{mN\times 1}^{{{i}_{1}}},\mathbf{c}_{mN\times 1}^{{{i}_{2}}} \right\rangle  \right|\le \frac{1}{N}\left( {{d}_{\mathbf{P}}}+\left( N-{{d}_{\mathbf{P}}} \right){{\mu }_{\mathbf{A}}} \right)\end{equation}
with
\begin{equation}\label{Eqn:---}{{\mu }_{\mathbf{C}}}=\frac{1}{N}\left( {{d}_{\mathbf{P}}}+\left( N-{{d}_{\mathbf{P}}} \right){{\mu }_{\mathbf{A}}} \right)\end{equation}
In this way, we have achieved a new deterministic sensing matrix of coherence $\frac{1}{N}\left( {{d}_{\mathbf{P}}}+\left( N-{{d}_{\mathbf{P}}} \right){{\mu }_{\mathbf{A}}} \right)$ based on the primary sensing matrix ${{\mathbf{A}}_{m\times n}}$.


\begin {example} \label{Euler}
let the identity matrix $\mathbf{A}$ of size $q \times q$ be a sampling matrix with zero cohrerence, where $q$ is any prime power. Moreover, consider the matrix $\mathbf P$  of size $k \times q^2$ which is constructed as follows. Firstly, deploy all the codewords generated by the generator matrix ${{\mathbf{G}}_{2\times q}}={{\left[ {{\mathbf{g}}_{0}},{{\mathbf{g}}_{1}} \right]}^{T}}$ over $\mathbb F_q$ where ${{\mathbf{g}}_{0}}=\left[ 1,\cdots ,1 \right]_{1\times q}^{T}$ and ${{\mathbf{g}}_{1}}=\left[ \lambda_0, \cdots ,\lambda_{q-1} \right]_{1\times q}^{T}$, and $\lambda_i \in \mathbb{F}_q, i = 0,\cdots, q-1$, as the columns of matrix $\mathbf P_0$. Then, arbitrarily, choose $k$ rows of the matrix $\mathbf P_0$ to construct $\mathbf P$. Trivially, for this matrix, $d_\mathbf P$ is $1$, and therefore, through applying Theorem \ref{CR_Resizing_Method}, we will obtain binary sensing matrices of size $kq \times q^2$ with coherence $\frac 1k$. The result of this example resembles that of \cite{naidu_14}, and hence, shows the generality of this approach.
\end {example}

\begin {example} \label{125_3125}

  Let \\ \scalebox{.8}[1]{$\mathbf{A} = \left[ {\mathbf{a}_{p^2  \times 1}^0  \cdots \mathbf{a}_{p^2  \times 1}^{\gamma _2 \alpha ^2  + \gamma _1 \alpha ^1  + \gamma _0 }  \cdots } \right.\\ \left. {\mathbf{a}_{p^2  \times 1}^{\left( {p - 1} \right)\alpha ^2  + \left( {p - 1} \right)\alpha ^1  + \left( {p - 1} \right)} } \right]$} be a matrix of size ${{p}^{2}}\times {{p}^{3}}$ and coherence ${1}/{p}\;$, where ${{\gamma }_{0}},{{\gamma }_{1}},{{\gamma }_{2}}\in {{\mathbb{F}}_{p}}$, and ${{\mathbf{P}}_{p\times {{p}^{6}}}}$ be a matrix whose columns are codewords over ${{\mathbb{F}}_{{{p}^{3}}}}$ generated by the generator matrix ${{\mathbf{G}}_{2\times p}}={{\left[ {{\mathbf{g}}_{0}},{{\mathbf{g}}_{1}} \right]}^{T}}$ where ${{\mathbf{g}}_{0}}=\left[ 1,\cdots ,1 \right]_{1\times p}^{T}$ and ${{\mathbf{g}}_{1}}=\left[ 0,\cdots ,p-1 \right]_{1\times p}^{T}$. Obviously, the codewords of ${{\mathbf{P}}_{p\times {{p}^{6}}}}$ have a minimum distance equal to $p-1$,  and therefore, each two arbitrary columns of ${{\mathbf{P}}_{p\times {{p}^{6}}}}$ have at most one element in common; i.e. ${{d}_{\mathbf{P}}}=1$. Consequently, the sensing matrix $\mathbf{C}$ which can be resulted from Theorem \ref{CR_Resizing_Method} based on ${{\mathbf{A}}_{{{p}^{2}}\times {{p}^{3}}}}$ and ${{\mathbf{P}}_{p\times {{p}^{6}}}}$ is a matrix of size ${{p}^{3}}\times {{p}^{6}}$ and coherence ${{\mu }_{\mathbf{C}}}={\left( 2p-1 \right)}/{{{p}^{2}}}\;$.
For comparison purposes, we assume that the size of Kronecker product matrix $\mathbf{{C}'}$ is ${{p}^{3}}\times {{p}^{6}}$. According to Theorem \ref{Kronecker_Resizing_Method},  for  ${{\mathbf{A}}_{{{p}^{2}}\times {{p}^{3}}}}$ we need a matrix of size $p\times {{p}^{3}}$ to reach the same size of  $\mathbf{{C}'}$. In other words, $\mathbf{B}$ is a sensing matrix of size $p\times {{p}^{3}}$ whose coherence according to the Welch bound cannot be less than ${{\mu }_{\mathbf{B}}}=\sqrt{\frac{p+1}{{{p}^{2}}+p+1}}$. Therefore,
\begin{equation}\label{Eqn:---}{{\mu }_{{\mathbf{{C}'}}}}=\max \left( {{\mu }_{\mathbf{A}}},{{\mu }_{\mathbf{B}}} \right)=\max \left( \frac{1}{p},\sqrt{\frac{p+1}{{{p}^{2}}+p+1}} \right)\end{equation}

Consider the following comparison,
\begin{equation}\label{Eqn:---}\sqrt{\frac{p+1}{{{p}^{2}}+p+1}}\overset{?}{\mathop{>}}\,\frac{1}{p}\ \ \ \Rightarrow \ \ {{p}^{3}}\overset{?}{\mathop{>}}\,p+1\end{equation}
Obviously, ${{p}^{3}}>p+1$ for any prime integer, and therefore,
\begin{equation}\label{Eqn:---}{{\mu }_{{\mathbf{{C}'}}}}=\sqrt{\frac{p+1}{{{p}^{2}}+p+1}}\end{equation}

Now, we compare the coherence of the proposed sensing matrix obtained through resizing by the column replacement method and that of \cite{Amini_19}. Therefore, comparison between ${{\mu }_{{\mathbf{{C}'}}}}$ and ${{\mu }_{\mathbf{C}}}$ results in

\begin{equation}\label{Eqn:Coh_Comparison}
\resizebox{1\hsize}{!}{$
{\left( 2p-1 \right)}/{{{p}^{2}}}\;\overset{?}{\mathop{<}}\,\sqrt{\frac{p+1}{{{p}^{2}}+p+1}}\ \Rightarrow \ 0\overset{?}{\mathop{<}}\,\left( {{p}^{4}}-p \right)\left( p-3 \right)+1$}
\end{equation}

From (\ref{Eqn:Coh_Comparison}), one can easily see that for the prime integers more than 3, the proposed method results in a lower coherence.
Note that we assumed $\mathbf{B}$ is a sensing matrix of size $p\times {{p}^{3}}$ attaining the Welch bound. However, to the best of our knowledge, the best known deterministic sensing matrices asymptotically attaining the Welch bound are of size $p\times {{p}^{2}}$. This result again corroborates our method as a better approach compared to the Kronecker product method in \cite{Amini_19}.
\end{example}

\section {Simulation results} \label{Simulation_results}

We compare the proposed method with those of  \cite{Amini_19} and random sensing matrices. To do so, we consider three different scenarios to verify all the examples presented in Section III. The results are shown based on the average of 3000 independent trials (1000 trials for Figs. 3 and 4) for different $k$-sparse signals. In each scenario, the measurement matrix performance is verified based on the recovery percentage as well as the output SNR, while the sparsity order is increasing and the input SNR is augmented from 0 to 100 dB. The reconstruction algorithm is Orthogonal Matching Pursuit which is suitable to solve ${{l}_{1}}$-minimization problems.
In the first scenario, we examine Example 1 in which matrices of size ${{p}^{2}}\times {{p}^{3}}$ and coherence ${1}/{p}\;$ are generated over finite field ${{\mathbb{F}}_{5}}$. We get a sensing matrix of size $25\times 125$ and coherence of ${1}/{5}\;$.  Over finite field ${{\mathbb{F}}_{5}}$, the construction of \cite{Amini_19} leads to a size of $24\times 125$ with coherence of ${1}/{4}\;$.  The third matrix in this scenario is a Gaussian sensing matrix of size $25\times 125$. The results are depicted in Figs. 1 and 2. Sparsity order of 12 has been considered in Fig.2. Apparently, our proposed deterministic sensing matrices act better than that of \cite{Amini_19}.
Meanwhile, both deterministic structures surpass random sensing matrices in the sense of recovery percentage as well as the output SNR.

\begin{figure*}[!ht]
\centering
\begin{minipage}[b]{0.4\textwidth}
    \includegraphics[width=\textwidth]{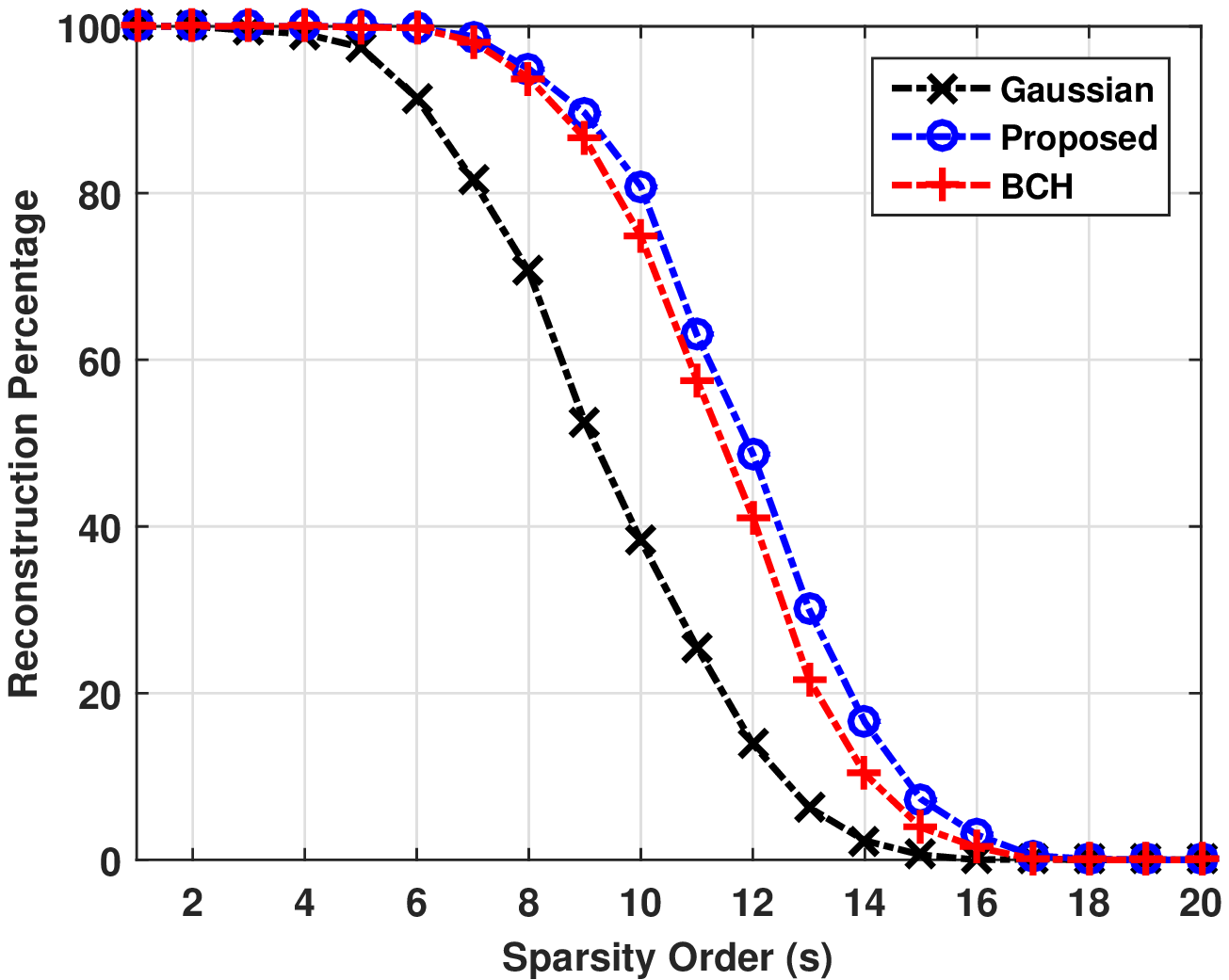}
  \caption{\label{Recons_Spars_42} Reconstruction percentage versus sparsity order.}
\end{minipage}
\hfill
  \centering
\begin{minipage}[b]{0.4\textwidth}
    \includegraphics[width=\textwidth]{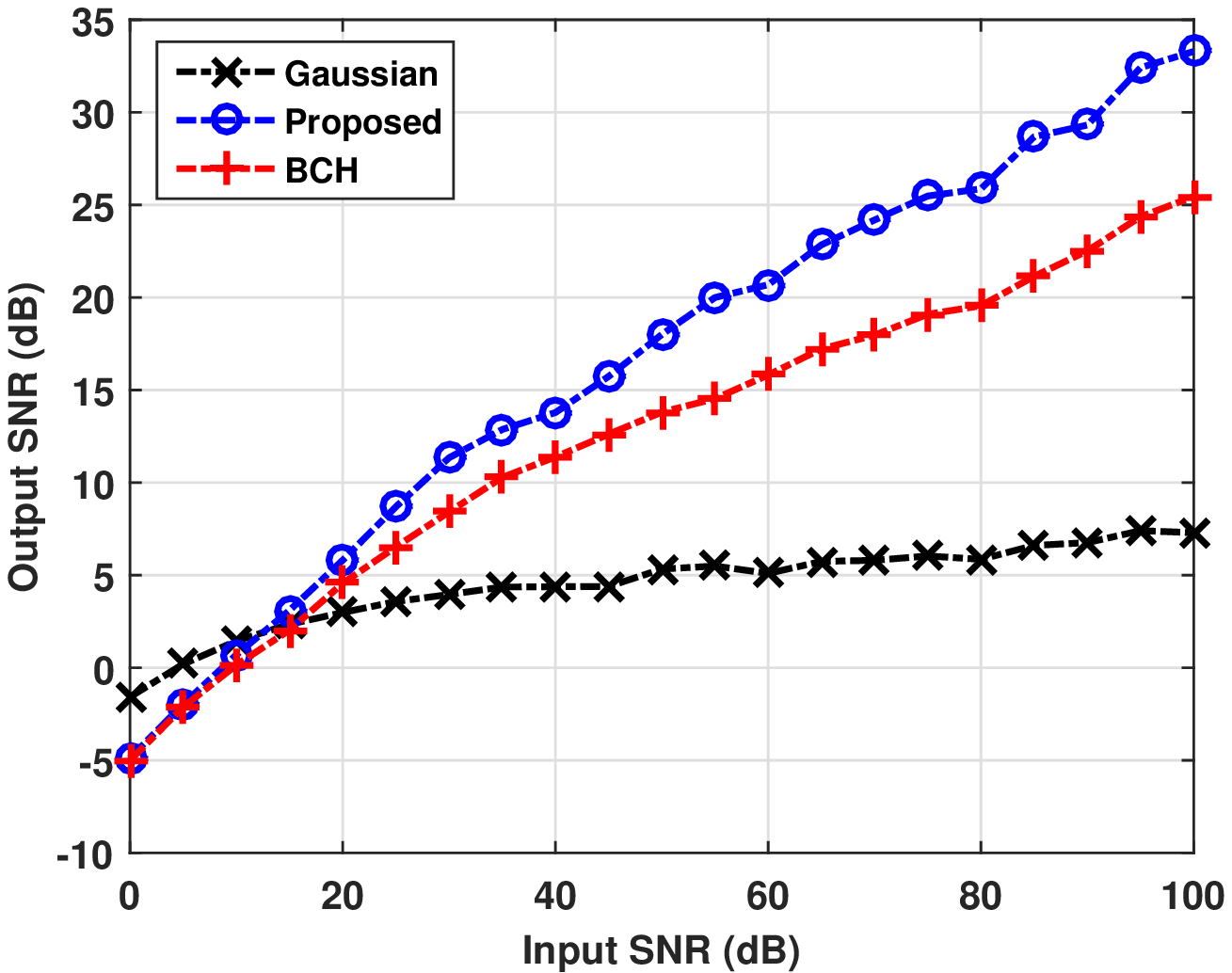}
  \caption{\label{SNR_SNR_42}Reconstruction SNR versus input SNR; sparsity of 13.}
\end{minipage}
\end{figure*}

In the last scenario, we compare our method to resize existing sensing matrices with the Kronecker product method. We assume that the primary sensing matrix $\mathbf{A}$ is of size $25\times 125$ and coherence ${1}/{5}\;$ and according to our example we need a pattern matrix of size $5\times 15625$ for each two distinct columns of which there exists at most one element in common and the resulting sensing matrix will be of size $125\times 15625$ and coherence $0.36$. On the other hand, to obtain a similar size sensing matrix based on the Kronecker product method, we need a matrix of size $5\times 125$ with near optimal coherence. Such a matrix has not been introduced so far and the best known structure attaining the Welch bound with the same number of rows is a matrix of size $5\times 25$ and coherence of ${1}/{\sqrt{5}}\;$. Therefore, we resize $\mathbf{A}$ up to $125\times 3125$ for which the resulting matrix has a coherence of  ${1}/{\sqrt{5}}\;$. Although, it is severely unfair to our construction; in which we could obtain a matrix of size $125\times 15625$ and coherence $0.36$, for comparison with the Kronecker product method, we have to reduce its size to $125\times 3125$. We also consider a Gaussian sensing matrix of size $125\times 3125$.
Finally, Figs. 3 and 4, where the sparsity is fixed to 28 to generate Fig. 5, illustrate the results of the third scenario in which we proposed a new method for resizing deterministic sensing matrices. As seen, in this example our method performs far better than the Kronecker product method. Notice that we have only kept ${1}/{5}\;$ of the available columns in our structure which reveals another advantage of the proposed method.

\begin{figure*}[!ht]
\centering
\begin{minipage}[b]{0.4\textwidth}
    \includegraphics[width=\textwidth]{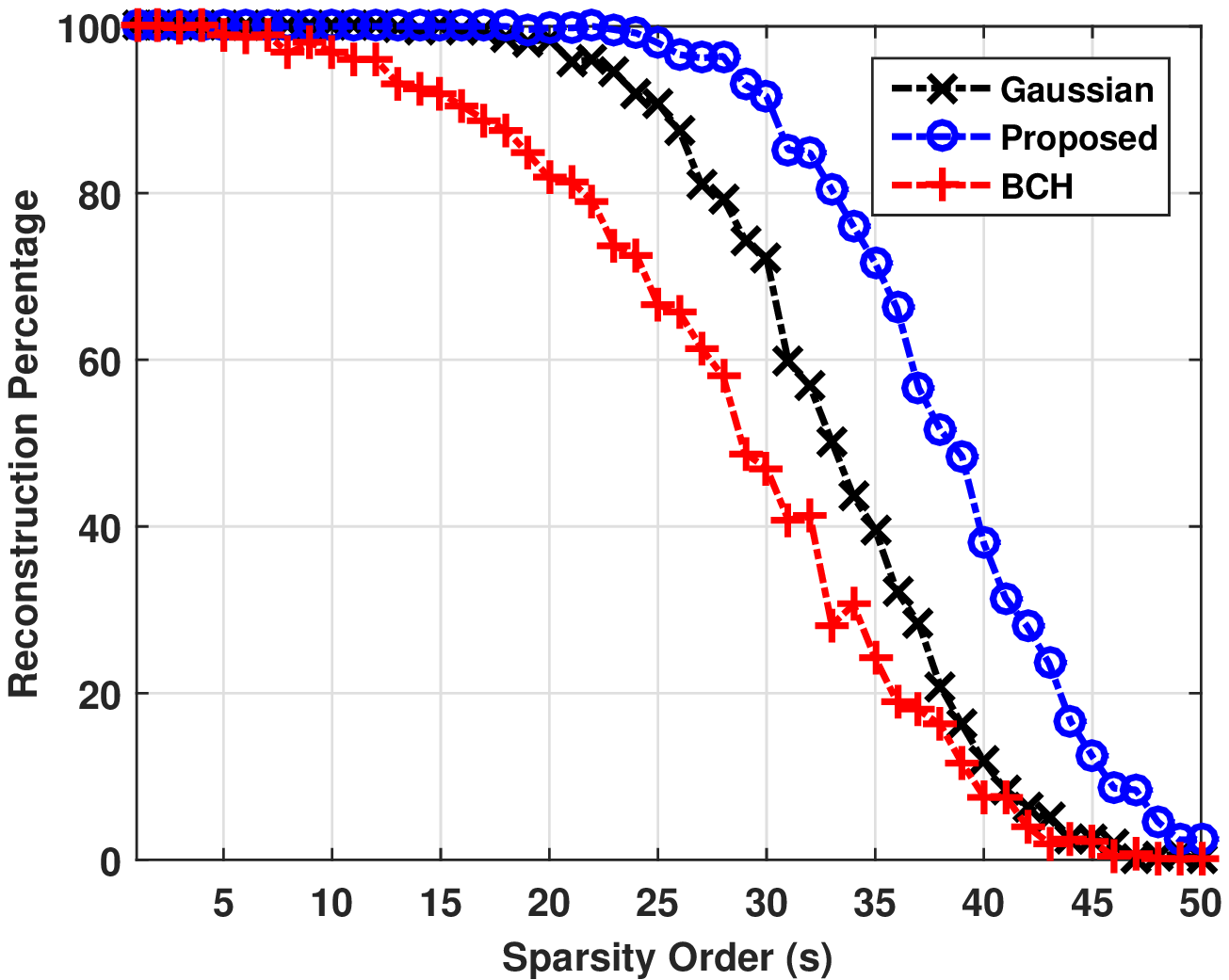}
  \caption{\label{Recons_Spars_125}Reconstruction percentage versus sparsity order.}
\end{minipage}
\hfill
  \centering
\begin{minipage}[b]{0.4\textwidth}
    \includegraphics[width=\textwidth]{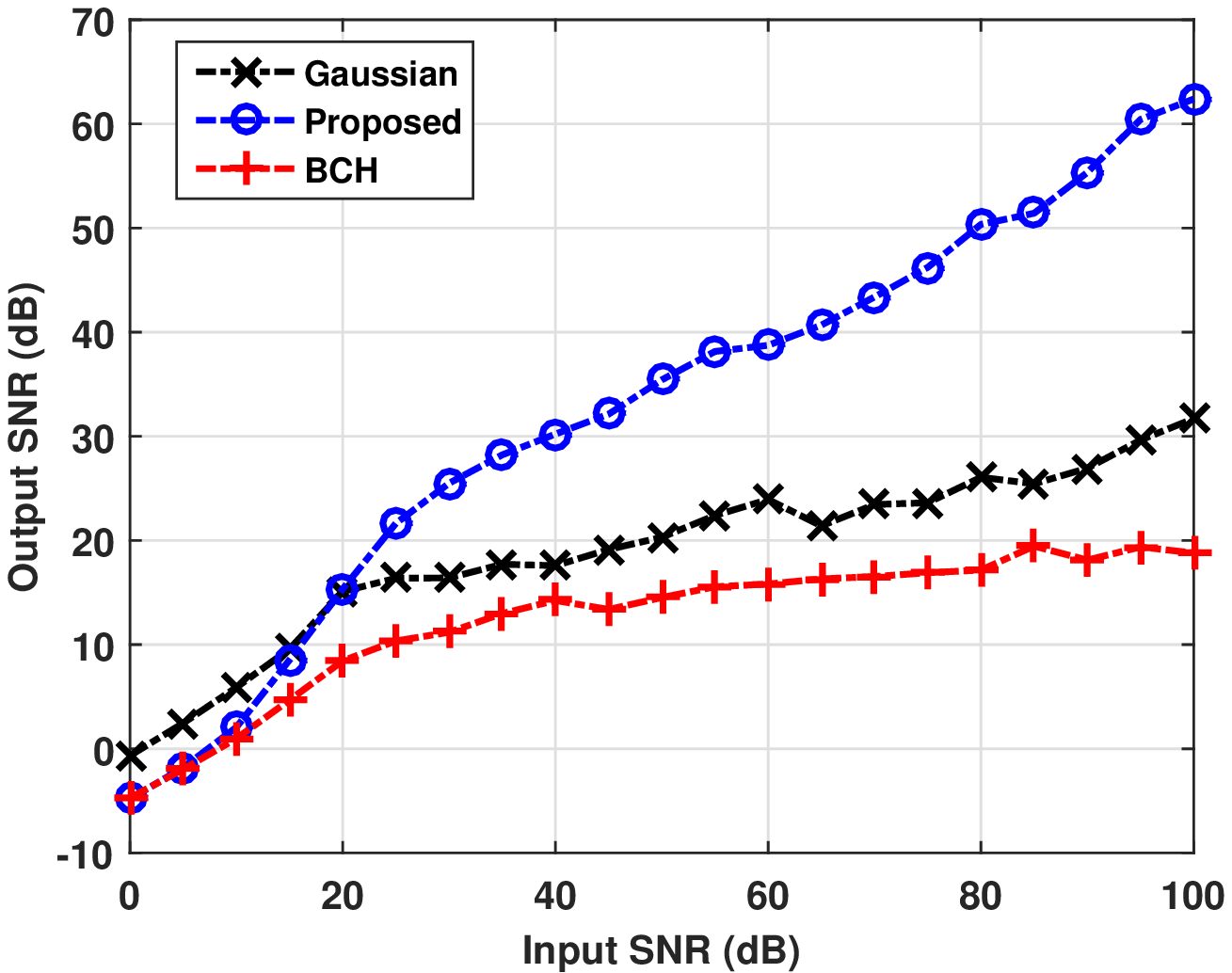}
  \caption{\label{SNR_SNR_125}Reconstruction SNR versus input SNR; sparsity of 39.}
  \centering
\end{minipage}
\end{figure*}

\section {Conclusion} \label{Conclusion}

 The column replacement concept was used to construct low coherence deterministic sensing matrices for CS problems.
 First, this method was employed to create deterministic sensing matrices. It was applied over two matrices of
 linear codewords to generate a new matrix of linear codewords suitable for creating sensing matrices. By using
 two examples and comparing with some existing asymptotically optimal matrices, we illustrated the generality of
 the proposed method. Secondly, we exploited the column replacement method to resize existing sensing matrices
 and found a bound for the coherence of resulted matrices. Using an example, as an especial case, it was shown
 that this method behaves better than the Kronecker product method developed in \cite{Amini_19}.
Observe that only a special case of Theorem \ref{Our_CR_Method} results in the matrices proposed in \cite{Mohades_15}, and hence, one may find much more diverse and
applicable structures by manipulating Theorem \ref{Our_CR_Method}.
Moreover, as seen in Example \ref{Euler}, our resizing approach is not only a way to increase the size of sampling matrices but also an approach to obtain well-structured sensing matrices.

 Finally, simulation results showed that by utilizing the proposed column
replacement method, the constructed
matrices perform better than the
existing deterministic sensing matrices
with approximately the same size. Moreover, simulation
results confirmed that the suggested
procedure of resizing sensing matrices
can perform much better compared to the
Kronecker product method.

\renewcommand{\theequation}{A-\arabic{equation}}
  \setcounter{equation}{0}  
  \section*{APPENDIX A} \label{APP_A} 

Consider the coherence as follows.

\begin{equation}\label{Eqn:Coh_APP_A}
{
{{\mu }_{{\mathbf{{C}''}}}}=\underset{{{{{\lambda }''_{01}}}},{{{{\lambda }''_{10}}}},{{{{\lambda }''_{11}}}}\in {{\mathbb{F}}_{p}}}{\mathop{\max }}\,\frac{1}{{{p}^{2}}}\left| \sum\limits_{\phi =0}^{p-1}{\sum\limits_{\gamma =0}^{p-1}{{{e}^{j\frac{2\pi }{p}\left( \gamma \left( {{{{\lambda }''_{11}}}}\phi +{{{{\lambda }''_{01}}}} \right)+{{{{\lambda }''_{10}}}}\phi  \right)}}}} \right|}
\end{equation}

Contemplating (\ref{Eqn:Additive_Char}), the inner summation of  (\ref{Eqn:Coh_APP_A}) is non-zero, iff:
\begin{equation}\label{Eqn:---}{{{\lambda }''_{11}}}\phi +{{{\lambda }''_{01}}}=0\end{equation}

Then, the following cases will happen:
\\
1-  $\left( {{{{\lambda }''}}_{11}},{{{{\lambda }''}}_{01}} \right)=\left( 0,0 \right)$ and ${{{\lambda }''_{10}}}\ne 0$, for which  (\ref{Eqn:Shrinked_equation_Coh}) reduces to
\begin{equation}\label{Eqn:---}{{\mu }_{{\mathbf{{C}''}}}}=\underset{\left( {{{{\lambda }''}}_{10}}\ne 0 \right)\in {{\mathbb{F}}_{p}}}{\mathop{\max }}\,\frac{1}{p}\left| \underbrace{\sum\limits_{\phi =0}^{p-1}{\underbrace{{{e}^{j\frac{2\pi }{p}\left( {{{{\lambda }''}}_{10}}\phi  \right)}}}_{\chi \left( {{{{\lambda }''}}_{10}}\phi  \right)}}}_{=0} \right|=0\end{equation}

In this case, due to the non-zero value of  $\lambda ''_{10} $,  coherence is zero.
\\
2-  ${{{\lambda }''_{11}}}\ne 0$ and ${{{\lambda }''_{10}}},{{{\lambda }''_{01}}}$ are arbitrarily chosen. Consequently, ${{{\lambda }''_{11}}}\phi +{{{\lambda }''_{01}}}$ is zero only for one value of $\phi $, say $\phi ={{\phi }_{1}}$. Therefore, the coherence value is given by
\renewcommand{\theequation}{A-\arabic{equation}}

\begin{equation}\label{Eqn:---}
\begin{gathered}
  {\mu _{{\mathbf{C''}}}} =  \hfill \\
  \mathop {\max }\limits_{\left\{ {\begin{array}{*{20}{c}}
  \hspace{-.2cm}\begin{subarray}{l} 
  0 \ne {{\lambda ''_{11}}}, \\ 
  {{\lambda ''_{01}}}, \\ 
  {{\lambda ''_{10}}} 
\end{subarray} &\hspace{-.3cm}{ \in {\mathbb{F}_p}} 
\end{array}} \right.} {\mkern 1mu}\hspace{-.3cm} \frac{1}{{{p^2}}}\left( {\left| {\underbrace {\sum\limits_{\left( {\phi  \ne {\phi _1}} \right) \in {\mathbb{F}_p}}^{} {\underbrace {\sum\limits_{\gamma  = 0}^{p - 1} {\underbrace {{e^{j\frac{{2\pi }}{p}\left( {\gamma \left( {{{\lambda ''_{11}}}\phi  + {{\lambda ''_{01}}}} \right) + {{\lambda ''_{10}}}\phi } \right)}}}_{\chi \left( {\gamma \left( {{{\lambda ''_{11}}}\phi  + {{\lambda ''_{01}}}} \right) + {{\lambda ''_{10}}}\phi } \right)}} }_{ = 0}} }_{ = 0} + } \right.} \right. \hfill \\ \hspace{+0.5cm}
  \,\,\,\,\,\,\,\,\,\,\,\,\,\,\,\,\,\,\,\,\,\,\,\,\,\,\,\,\,\,\,\,\left. {\left. {\underbrace {\sum\limits_{\gamma  = 0}^{p - 1} {\underbrace {{e^{j\frac{{2\pi }}{p}\left( {\gamma \underbrace {\left( {{{\lambda ''_{11}}}{\phi _1} + {{\lambda ''_{01}}}} \right)}_{ = 0} + {{\lambda ''_{10}}}{\phi _1}} \right)}}}_{\chi \left( {{{\lambda ''_{10}}}{\phi _1}} \right)}} }_{p*{e^{j\frac{{2\pi }}{p}\left( {{{\lambda ''_{10}}}{\phi _1}} \right)}}}} \right|} \right) = \frac{1}{p} \hfill \\ 
\end{gathered}
\end{equation}
From the above cases, we can deduce that the coherence of  $\mathbf{{C}''}$is ${1}/{p}\;$, i.e. ${{\mu }_{{\mathbf{{C}''}}}}={1}/{p}\;$.

\renewcommand{\theequation}{B-\arabic{equation}}
  \setcounter{equation}{0}  
  \section*{APPENDIX B}  

With the same policy adopted in Appendix A, we obtain the coherence of ${{\mathbf{{C}''}}_{p\left( p-1 \right)\times {{p}^{3}}}}$ as
\begin{equation}\label{Eqn:---}
\resizebox{.72 \hsize}{!}{$
{{\mu }_{{\mathbf{{C}''}}}}=\underset{j\ne k}{\mathop{\max }}\,\left| \left\langle {{{\mathbf{{c}''}}}^{j}}_{p\left( p-1 \right)\times 1},{{{\mathbf{{c}''}}}^{k}}_{p\left( p-1 \right)\times 1} \right\rangle  \right|$}
\end{equation}
\begin{equation}\label{Eqn:Coh_APP_B}
\begin{gathered}
  {\mu _{{\mathbf{C''}}}} =  \hfill \\
  \mathop {\max }\limits_{\left\{ {\begin{array}{*{20}{c}}
\hspace{-.2cm}  \begin{subarray}{l} 
  {{\lambda ''_{01}}}, \\ 
  {{\lambda ''_{10}}}, \\ 
  {{\lambda ''_{11}}} 
\end{subarray} &\hspace{-.3cm}{ \in {\mathbb{F}_p}} 
\end{array}} \right.} {\mkern 1mu} \frac{1}{{p\left( {p - 1} \right)}}\left| {\sum\limits_{\phi  = 0}^{p - 2} {\sum\limits_{\gamma  = 0}^{p - 1} {{e^{j\frac{{2\pi }}{p}\left( {\gamma \left( {{{\lambda ''_{11}}}\phi  + {{\lambda ''_{01}}}} \right) + {{\lambda ''_{10}}}\phi } \right)}}} } } \right| \hfill \\ 
\end{gathered} 
\end{equation}
Considering (\ref{Eqn:Additive_Char}), the inner summation of (\ref{Eqn:Coh_APP_B}) is non-zero, iff:
\begin{equation}\label{Eqn:---}{{{\lambda }''_{11}}}\phi +{{{\lambda }''_{01}}}=0\end{equation}

Then, the following cases will happen:
\\
\\
1-  ${{{\lambda }''_{11}}}\ne 0$ and ${{{\lambda }''_{10}}},{{{\lambda }''_{01}}}$ are arbitrarily chosen. Consequently, ${{{\lambda }''_{11}}}\phi +{{{\lambda }''_{01}}}$ might be zero only for one value of $\phi $, say $\phi ={{\phi }_{1}}\ne p-1$. Note that if the only choice which sets ${{{\lambda }''_{11}}}\phi +{{{\lambda }''_{01}}}$ zero, is $\phi =p-1$, then, according to (\ref{Eqn:Additive_Char}), ${{\mu }_{{\mathbf{{C}''}}}}=0$. This is due to previously considering $\phi \in \left\{ {{\mathbb{F}}_{p}}-\left( p-1 \right) \right\}$).
Then, the coherence at the worst case is

\renewcommand{\theequation}{B-\arabic{equation}}

\begin{equation}\label{Eqn:---}
\resizebox {1 \hsize} {!}{$
\begin{gathered}
  {\mu _{{\mathbf{C''}}}} =  \hfill \\
 \hspace{-.2cm} \mathop {\max }\limits_{\left\{ {\begin{array}{*{20}{c}}
  \hspace{-.2cm}\begin{subarray}{l} 
  0 \ne {{\lambda ''_{11}}}, \\ 
  {{\lambda ''_{01}}}, \\ 
  {{\lambda ''_{10}}} 
\end{subarray} &\hspace{-.4cm}{ \in {\mathbb{F}_p}} 
\end{array}} \right.}\hspace{-.25cm} \frac{1}{{p\left( {p - 1} \right)}}\left( {\left| {\underbrace {\sum\limits_{\left( {\phi  \ne {\phi _1}} \right) \in \left\{ {{\mathbb{F}_p} - \left( {p - 1} \right)} \right\}}^{} {\underbrace {\sum\limits_{\gamma  = 0}^{p - 1} {\underbrace {{e^{j\frac{{2\pi }}{p}\left( {\gamma \left( {{{\lambda ''_{11}}}\phi  + {{\lambda ''_{01}}}} \right) + {{\lambda ''_{10}}}\phi } \right)}}}_{\chi \left( {\gamma \left( {{{\lambda ''_{11}}}\phi  + {{\lambda ''_{01}}}} \right) + {{\lambda ''_{10}}}\phi } \right)}} }_{ = 0}} }_{ = 0} + } \right.} \right. \hfill \\ \hspace{+2.9cm}
  \left. {\left. {\underbrace {\sum\limits_{\gamma  = 0}^{p - 1} {\underbrace {{e^{j\frac{{2\pi }}{p}\left( {\gamma \underbrace {\left( {{{\lambda ''_{11}}}{\phi _1} + {{\lambda ''_{01}}}} \right)}_{ = 0} + {{\lambda ''_{10}}}{\phi _1}} \right)}}}_{\chi \left( {{{\lambda ''_{10}}}{\phi _1}} \right)}} }_{p*{e^{j\frac{{2\pi }}{p}\left( {{{\lambda ''_{10}}}{\phi _1}} \right)}}}} \right|} \right) = \frac{1}{{p - 1}} \hfill \\ 
\end{gathered}$} \end{equation}
\\
2-  $\left( {{{{\lambda }''}}_{11}},{{{{\lambda }''}}_{01}} \right)=\left( 0,0 \right)$ and ${{{\lambda }''_{10}}}\ne 0$ for which we get
\begin{equation}\label{Eqn:---}
\begin{gathered}
  {\mu _{{\mathbf{C''}}}} =  \hfill \\
  \,\,\mathop {\max }\limits_{\left( {{{\lambda ''_{10}}} \ne 0} \right) \in {\mathbb{F}_p}} {\mkern 1mu} \frac{1}{{p - 1}}\left| {\underbrace {\sum\limits_{\phi  = 0}^{p - 1} {\underbrace {{e^{j\frac{{2\pi }}{p}\left( {{{\lambda ''_{10}}}\phi } \right)}}}_{\chi \left( {{{\lambda ''_{10}}}\phi } \right)}} }_{ = 0} - {e^{j\frac{{2\pi }}{p}\left( {{{\lambda ''_{10}}}\left( {p - 1} \right)} \right)}}} \right| \hfill \\
  \,\,\, = \frac{1}{{p - 1}} \hfill \\ 
\end{gathered} 
\end{equation}
Based on the above cases, one can infer that the coherence of $\mathbf{{C}''}$ is ${1}/{\left( p-1 \right)}\;$, i.e. ${{\mu }_{{\mathbf{{C}''}}}}={1}/{\left( p-1 \right)}\;$.



\begin{thebibliography}{9}

\bibitem{Donoho_1}
D.~L. Donoho, ``Compressed sensing,'' {\em Information Theory, IEEE
  Transactions on}, vol.~52, no.~4, pp.~1289--1306, 2006.

\bibitem{Candes_2}
E.~J. Candes and T.~Tao, ``Near-optimal signal recovery from random
  projections: Universal encoding strategies?,'' {\em Information Theory, IEEE
  Transactions on}, vol.~52, no.~12, pp.~5406--5425, 2006.

\bibitem{Ge_3}
D.~Ge, X.~Jiang, and Y.~Ye, ``A note on the complexity of l p minimization,''
  {\em Mathematical programming}, vol.~129, no.~2, pp.~285--299, 2011.

\bibitem{Candes_4}
E.~J. Candes and T.~Tao, ``Decoding by linear programming,'' {\em Information
  Theory, IEEE Transactions on}, vol.~51, no.~12, pp.~4203--4215, 2005.

\bibitem{Tropp_5}
J.~A. Tropp and A.~C. Gilbert, ``Signal recovery from random measurements via
  orthogonal matching pursuit,'' {\em Information Theory, IEEE Transactions
  on}, vol.~53, no.~12, pp.~4655--4666, 2007.

\bibitem{Baraniuk_6}
R.~Baraniuk, M.~Davenport, R.~DeVore, and M.~Wakin, ``A simple proof of the
  restricted isometry property for random matrices,'' {\em Constructive
  Approximation}, vol.~28, no.~3, pp.~253--263, 2008.

\bibitem{Bourgain_7}
J.~Bourgain, S.~Dilworth, K.~Ford, S.~Konyagin, D.~Kutzarova, {\em et~al.},
  ``Explicit constructions of rip matrices and related problems,'' {\em Duke
  Mathematical Journal}, vol.~159, no.~1, pp.~145--185, 2011.

\bibitem{Welch_8}
L.~R. Welch, ``Lower bounds on the maximum cross correlation of signals
  (corresp.),'' {\em Information Theory, IEEE Transactions on}, vol.~20, no.~3,
  pp.~397--399, 1974.

\bibitem{DeVore_9}
R.~A. DeVore, ``Deterministic constructions of compressed sensing matrices,''
  {\em Journal of Complexity}, vol.~23, no.~4, pp.~918--925, 2007.

\bibitem{Mohades_10}
A.~Mohades and A.~Tadaion, ``Finite projective spaces in deterministic
  construction of measurement matrices,'' {\em IET Signal Processing}, 2016.

\bibitem{Li_11}
S.~Li, F.~Gao, G.~Ge, and S.~Zhang, ``Deterministic construction of compressed
  sensing matrices via algebraic curves,'' {\em Information Theory, IEEE
  Transactions on}, vol.~58, no.~8, pp.~5035--5041, 2012.

\bibitem{sasmal_24}
P.~Sasmal, R.~R. Naidu, C.~S. Sastry, and P.~Jampana, ``Composition of binary
  compressed sensing matrices,'' {\em IEEE Signal Processing Letters}, vol.~23,
  no.~8, pp.~1096--1100, 2016.

\bibitem{naidu_14}
R.~R. Naidu, P.~Jampana, and C.~S. Sastry, ``Deterministic compressed sensing
  matrices: Construction via euler squares and applications.,'' {\em IEEE
  Trans. Signal Processing}, vol.~64, no.~14, pp.~3566--3575, 2016.

\bibitem{naidu_25}
R.~R. Naidu and C.~R. Murthy, ``Construction of binary sensing matrices using
  extremal set theory,'' {\em IEEE Signal Processing Letters}, vol.~24, no.~2,
  pp.~211--215, 2017.

\bibitem{Mohades_15}
M.~Mohades, A.~Mohades, and A.~Tadaion, ``A reed-solomon code based measurement
  matrix with small coherence,'' {\em Signal Processing Letters, IEEE},
  vol.~21, no.~7, pp.~839--843, 2014.

\bibitem{Amini_16}
A.~Amini, H.~Bagh-Sheikhi, and F.~Marvasti, ``From paley graphs to
  deterministic sensing matrices with real-valued gramians,'' in {\em Sampling
  Theory and Applications (SampTA), 2015 International Conference on},
  pp.~372--376, IEEE, 2015.

\bibitem{Amini_18}
A.~Amini and F.~Marvasti, ``Deterministic construction of binary, bipolar, and
  ternary compressed sensing matrices,'' {\em Information Theory, IEEE
  Transactions on}, vol.~57, no.~4, pp.~2360--2370, 2011.

\bibitem{Amini_19}
A.~Amini, V.~Montazerhodjat, and F.~Marvasti, ``Matrices with small coherence
  using-ary block codes,'' {\em Signal Processing, IEEE Transactions on},
  vol.~60, no.~1, pp.~172--181, 2012.

\bibitem{Xu_20}
W.~Xu and B.~Hassibi, ``Efficient compressive sensing with deterministic
  guarantees using expander graphs,'' in {\em Information Theory Workshop,
  2007. ITW'07. IEEE}, pp.~414--419, IEEE, 2007.

\bibitem{Calderbank_21}
R.~Calderbank, S.~Howard, and S.~Jafarpour, ``Construction of a large class of
  deterministic sensing matrices that satisfy a statistical isometry
  property,'' {\em Selected Topics in Signal Processing, IEEE Journal of},
  vol.~4, no.~2, pp.~358--374, 2010.

\bibitem{Colbourn_23}
C.~J. Colbourn, D.~Horsley, and C.~McLean, ``Compressive sensing matrices and
  hash families,'' {\em Communications, IEEE Transactions on}, vol.~59, no.~7,
  pp.~1840--1845, 2011.

\bibitem{Finite_22}
R.~Lidl and H.~Niederreiter, {\em Finite fields},
\newblock (Cambridge university press, Cambridge, 2nd edn. 1997).

\end{thebibliography}
%





\end{document}